\documentclass[10pt]{article}
\usepackage{dcolumn}% Align table columns on decimal point
\usepackage{bm}% bold math
\usepackage{verbatim}       % Defines \begin{comment}, \end{comment}

\usepackage[dvips]{graphicx}

\usepackage{amsthm}
\usepackage{amssymb}

\usepackage{bm}
\usepackage{amstext}
\usepackage{amsmath}
\usepackage{amsfonts}
\usepackage{mathtools}
\usepackage{units}
\usepackage{mathrsfs}
\newcommand{\p}{\partial}

\newcommand{\dd}{{\rm d}}

\newcommand{\bd}{\begin{definition}}                %inizia definizione
\newcommand{\ed}{\end{definition}}                  %fine definizione
\newcommand{\bc}{\begin{corollary}}                 %inizia corollario
\newcommand{\ec}{\end{corollary}}                   %fine corollario
\newcommand{\bl}{\begin{lemma}}                     %inizia lemma
\newcommand{\el}{\end{lemma}}                       %fine lemma
\newcommand{\bp}{\begin{proposition}}            %inizia proposizione
\newcommand{\ep}{\end{proposition}}                %fine proposizione
\newcommand{\bere}{\begin{remark}}                  %inizia osservazione
\newcommand{\ere}{\end{remark}}                     %fine oservazione

\newcommand{\bt}{\begin{theorem}}
\newcommand{\et}{\end{theorem}}

\newcommand{\be}{\begin{equation}}
\newcommand{\ee}{\end{equation}}

\newcommand{\bit}{\begin{itemize}}
\newcommand{\eit}{\end{itemize}}
\newtheorem{theorem}{Theorem}[section]
\newtheorem{corollary}[theorem]{Corollary}
\newtheorem{lemma}[theorem]{Lemma}
\newtheorem{proposition}[theorem]{Proposition}
\theoremstyle{definition}
\newtheorem{definition}[theorem]{Definition}
\theoremstyle{remark}
\newtheorem{remark}[theorem]{Remark}

\begin{document}

\title{Relativistic Chasles' theorem and the conjugacy classes of the inhomogeneous Lorentz group}
%\title{An extension of Chasles' theorem to Minkowski space}

\author{E. Minguzzi\thanks{
Dipartimento di Matematica Applicata ``G. Sansone'', Universit\`a
degli Studi di Firenze, Via S. Marta 3,  I-50139 Firenze, Italy.
E-mail: ettore.minguzzi@unifi.it} }

\date{}

\maketitle

\begin{abstract}
\noindent This work is devoted to the relativistic generalization of
Chasles' theorem, namely to the proof that every proper
orthochronous isometry of Minkowski spacetime,  which sends some
point to its chronological future, is generated through the frame
displacement of an observer which moves with constant acceleration
and constant angular velocity. The acceleration and angular velocity
can be chosen either aligned or perpendicular, and in the latter
case the angular velocity can be chosen equal or smaller than the
acceleration. We start reviewing the classical Euler's and Chasles'
theorems both in the Lie algebra and group versions. We recall the
relativistic generalization of Euler's theorem and observe that
every (infinitesimal) transformation can be recovered from
information of algebraic and geometric type, the former being
identified with the conjugacy class and the latter with some
additional geometric ingredients (the screw axis in the usual
non-relativistic version). Then the proper orthochronous
inhomogeneous Lorentz Lie group is studied in detail. We prove its
exponentiality and identify a causal semigroup and the corresponding
Lie cone. Through the identification of new Ad-invariants we
classify the conjugacy classes, and show that those which admit a
causal representative have special physical significance. These
results imply a classification of the inequivalent Killing vector
fields of Minkowski spacetime which we express through simple
representatives. Finally, we arrive at the mentioned generalization
of Chasles' theorem.

\end{abstract}

%
%of physically allowed transformations
%
%
%
%exploring two characteristic timelike and causal subgroups.  The Lie
%orbits under the adjoint action of the group are classified, which
%implies a a classification of the equivalence classes of the
%Killings vector fields of Minkowski spacetime we we express through
%simple representatives.
%
%
%The proper orthochronous Lorentz transformations of Minkowski
%spacetime are studied in detail. The geometric and algebraic content
%of the transformation $P$ are identified. In particular, it is shown
%that it is possible to recover the transformation $P$ from this
%data, thus generalizing Chasles' theorem on rigid motions to the
%relativistic framework. The algebraic content is identified with the
%conjugacy class and a classification of conjugacy classes is given.
%Furthermore, a causal semigroup of physically allowed
%transformations is identified, and the conjugacy classes which admit
%a physically admissible representative are determined. The Lie wedge
%for the semigroup is also introduced and its role commented. The
%paper ends with some hints on how to construct a relativistic screw
%theory.

%
%It is proved that every Poincar\'e transformation, namely every
%ortochronous isometry of Minkowski spacetime, admits an invariant
%causal plane. In particular, the reference frame can be chosen in
%such a way that the endomorphism takes one of two possible forms.
%The case in which the invariant plane is timelike corresponds to the
%classical Chasles' theorem, the plane being the worldsheet of the
%instantaneous axis of rotation.

\tableofcontents

\section{Introduction}
A rigid movement is an orientation preserving isometry of Euclidean
space. A classical theorem by Euler states that every rigid movement
admitting a fixed point can be accomplished through a rotation
around some axis passing through the point. This result was
generalized by Mozzi and Chasles, \cite{ceccarelli00} who proved
%%%etsuko e ettore si vogliono bene
that in the general case in which no fixed point is required, the
rigid movement can be accomplished through a rotation around some
axis combined with a translation parallel to the axis. The
composition of these two movements can be accomplished with a single
screw or helical motion.

Mathematically, Euler's and Chasles' theorems establish the
existence of a certain type of representative for each conjugacy
class of the group $SO(3)$ and $ISO(3)$, respectively. The conjugacy
transformation represents a change of frame, thus at the geometrical
level the choice of a convenient representative corresponds to the
choice of a convenient frame.

A related  problem is that of finding the  orbits of the adjoint
(Ad) action of  $SO(3)$ on its Lie algebra $\mathfrak{so}(3)$. The
motivation is essentially the same: we wish to select a simple
element of the orbit on $\mathfrak{so}(3)$ so as to read with ease
the physical content of the infinitesimal transformation represented
by the Lie algebra element.
%Geometrically, the simplifying frame selects an axis called {\em the
%instantaneous axis of rotation}.
Usually the infinitesimal versions of Euler's and Chasles' theorem
are regarded as special cases of their finite counterparts. The
finite version can also be deduced from the infinitesimal one.  The
proof in this  direction is essentially equivalent to the proof that
the Lie group $SO(3)$ is exponential.

Ultimately, each Lie algebra element is a vector field and, in
Chasles' case, it can be represented with a characteristic screw flow
around a special line called  {\em instantaneous axis of rotation}.
The infinitesimal formulation of Chasles' theorem became the
starting point of {\em Screw Theory}, a formulation of rigid body
mechanics which unifies in the concept of screw the rotational and
translational degrees of freedom of rigid bodies
\cite{ball76,dimentberg68,murray94,selig05,minguzzi12}.

Euler's theorem was generalized to Minkowski space $M$ by several
authors
\cite{wigner39,abraham48,lomony63,schwartz63,schwartz65,shaw69}.
This problem is essentially equivalent to that of classifying the
conjugacy classes of the proper orthochronous Lorentz group.

%Using the push forward of the isometry $\Lambda$ of Minkowski
%spacetime we can send orthonormal bases to orthonormal bases. Once
%an initial base has been fixed, the coordinate transformation
%between the initial and final coordinate systems induced by those
%bases is mediated by a matrix $\Lambda^a_{\ b} \in
%SO(1,3)^\uparrow$. If we change the starting base the representing
%matrix changes through a conjugacy map $\Lambda^a_{\ b} \to L^c_{\
%a} \Lambda^a_{\ b} (L^{-1})^{b}_{\ d}$. The idea is to select a
%special representative for each conjugacy class in such a way that
%the physical interpretation of the Lorentz transformation $\Lambda$
%becomes particularly transparent.

%Another interesting approach uses the isomorphism between the
%Lorentz group and $PSL(2,\mathbb{C})$ (Lorentz transformations of
%the observer induce an action of $PSL(2,\mathbb{C})$ on the Riemann
%sphere of the 'night sky' \cite{penrose59,penrose84,naber92}). The
%conjugacy classes for $PSL(2,\mathbb{C})$ are well known
%\cite{needham97}.

%A similar analysis can be done for the inhomogeneous proper
%orthochronous Lorentz group.

In this work we  generalize Chasles' theorem by selecting a
convenient representative for each conjugacy class of the
inhomogeneous proper orthochronous Lorentz group. We identify the
type of geometric data which is required in order to recover the
original transformation. The simple form of the representatives will
simplify the interpretation and, in particular, will allow us to
prove a result which we can conveniently formulate as (we shall give
precise definitions of all the terms involved, see Theorems
\ref{vkw} and \ref{vtw})

\begin{theorem}
Every proper orthochronous  isometry of Minkowski spacetime, which
sends some point to its chronological future, can be accomplished
through the frame dragging of spacetime points, where the frame is
that of an observer which moves with constant angular velocity and
constant acceleration for some proper time interval. The observer
can be chosen so that the acceleration and angular velocity are
either aligned or perpendicular. In the latter case the angular
velocity can be chosen no greater than the acceleration.

Finally, there are two cases. If the observer's motion is of pure
rotation, the proper time interval of motion duration and the
angular velocity are uniquely determined, while if the observer's
motion cannot be chosen to be a pure rotation, then the proper time
interval can be chosen arbitrarily, and after this choice the
modules of the acceleration and angular velocity are uniquely
determined (Eqs. (\ref{cf1})-(\ref{cf3})).
\end{theorem}

%that every proper orthochronous  isometry of Minkowski
%spacetime, which sends some point to its chronological future, can
%be accomplished through the frame dragging of spacetime points,
%where the frame is that of an observer which moves with constant
%angular velocity and constant acceleration.

%We will also show that the angular velocity can be chosen to be
%aligned or orthogonal to the acceleration, and that in the latter
%case it can be chosen to be less or equal in module than the
%acceleration.

Up to the freedom in the time duration, the acceleration and angular
velocity are uniquely determined, thus they can be regarded as
genuine characteristics of the isometry. With respect to the
classical Chasles' theorem, here we need to impose a casuality
condition, indeed,  space translations are not generated by
observer's motions while they satisfy  the other hypothesis. As a
consequence, we shall need some results on the way causality
reflects itself on the Lie algebra. This will be done identifying a
causal  Lie semigroup and studying the corresponding Lie cone.

This paper is organized as follows.

In section 2 we recall the classical Euler's and Chasles' theorems,
both in the Lie group and Lie algebra versions. We notice here that
in order to recover the original transformation we need information
of algebraic and geometric type, the  former being identified with
the conjugacy class (Ad-orbit, in the Lie algebra case) and the
latter being identified with the screw axis. We also introduce the
screw product on the Lie algebra as its generalization will provide
a new Ad-invariant for the relativistic case.

In section 3 we study the Lorentz group introducing the usual
Ad-invariants for the Lie algebra, and recalling the  classification
of the Lie orbits and conjugacy classes. We also identify the
geometric data needed to recover the full (infinitesimal)
transformation.

In section 4 we come to the inhomogeneous Lorentz group. In section
4.1 we introduce a causal semigroup of $ISO(1,3)^\uparrow$, showing
its connection with isometries which send some point to its causal
future. In section 4.2 we introduce our conventions and clarify the
physical meaning of the Lie algebra generators. This section will be
essential for the correct interpretation of subsequent results. In
particular we explain the importance of linear combinations of the
form $\vec{a}\cdot \vec{K}+\vec{\omega}\cdot \vec{J}+H$, where the
translational generators $\vec{P}$ do not appear. Indeed, we
interpret these combinations as the allowed generators for the
observer's motion. The relativistic Chasles' theorem will ask not
only to prove that the generic transformation is the exponential of
some infinitesimal generator, a fact proved in section 4.3, but also
that the generator is of the mentioned form up to conjugacy.

In section 4.4 we  introduce a  set of Ad-invariants which allow us
to completely classify the orbits of the adjoint action of
$ISO(1,3)^\uparrow$, on $\mathfrak{iso}(1,3)$. This classification
implies a classification of the inequivalent Killing fields of
Minkowski spacetime. We clarify the relation between our
classification and a slightly coarser one previously obtained by T.
Barbot \cite{barbot05}.

In section 4.5 we introduce the Lie cone of the causal semigroup. We
answer the following question: given two frames (bases) in
spacetime, with the application point of the latter  in the
chronological future of the former, is it always possible to regard
them as the initial and final states of a comoving base attached to
an observer which rotates and accelerates with constant angular
velocity and acceleration for some proper time interval? The answer is negative unless the last
frame is contained in a spacetime cone which is narrower than the
light cone and which depends on the required angular velocity and
acceleration.

In section 4.6 we show that the Lie Ad-orbits can be given a causal
character depending on whether some representative belongs to the
causal Lie cone. The causal orbits are in a way reminiscent of the
classification of elementary particles (indeed, at least for finite
groups, there is a bijection between conjugacy classes and
irreducible linear representations). With theorem \ref{vkw} we
obtain the relativistic generalization of Chasles' theorem, in the
Lie algebra formulation. Finally, in section 4.7 we give the group
version.

Concerning our conventions, the indices $i,j,k$, take the values
$1,2,3$, while the indices $a,b,c$ or $\alpha,\beta,\gamma$, take
the values $0,1,2,3$.
%, and the indices $l,m,n$, take the values $2,3$.
We adopt the Einstein summation convention, and our signature for
the Minkowski metric $\eta_{ab}$ is $(-,+,+,+)$. A vector $v$ is
causal (timelike) if $\eta(v,v)\le 0$ (resp. $<0$) and $v\ne 0$. The
vector is nonspacelike if it is causal or $v=0$. A vector is
lightlike (null) if it is causal (resp. nonspacelike) but not
timelike. The chronological future $I^{+}(x)$ of $x\in M$ is
made by all the points that can be reached from $x$ following future
directed (f.d.) timelike curves. The causal future is denoted
$J^{+}(x)$ and includes $x$ plus all the point that can be reached
from $x$ following  f.d.\  causal curves. For shortness, we shall
sometimes use the word {\em direction} when referring to a
1-dimensional subspace of a vector space. We use units such that
$c=1$, where $c$ is the speed of light. The subset symbol $\subset$
is reflexive, i.e. $X \subset X$.

For background on the inhomogeneous Lorentz group the reader might consult \cite{carmeli77,weinberg95,sexl01}.

\section{Euler's and Chasles' theorems}

Let us formulate Chasles' theorem in  mathematical language. Let $E$
be the Euclidean space. This means that $E$ is an affine space
modeled over a 3-dimensional vector space $(V,\cdot, or)$, endowed
with a positive definite scalar product $\cdot: V\times V\to
\mathbb{R}$, and orientation $or$. A reference frame is a choice of
origin $o\in E$ plus a positive oriented orthonormal base $\{
\vec{e}_i, i=1,2,3\}$ of $V$. Given a reference frame, every point
$p\in E$ can be written in a unique way in terms of coordinates as
follows $p=o+x^i \vec{e_i}$. The coordinate vector belonging to
$\mathbb{R}^3$ will be denoted using a bar, e.g. $\bar{x}$.

The rigid motion $\psi: E\to E$ can be lifted to the bundle of
reference frames as follows: $(o, \{\vec{e}_i\}) \to (\psi(o),
\{\psi_{*}(\vec{e}_i)\})$. To this change of frame corresponds an
affine change of coordinates given by
%As it is well known, once a reference frame has been fixed, every
%rigid motion ,  can be written as an affine map $\bar{x}\to
%\bar{x}'$, as follows
\begin{equation} \label{nug}
\begin{pmatrix} \bar{x}' \\ 1 \end{pmatrix} =\begin{pmatrix} O & \bar{b} \\ \bar{0}^\intercal & 1 \end{pmatrix} \begin{pmatrix} \bar{x} \\ 1 \end{pmatrix}
\end{equation}
where $O$ is a special orthogonal matrix. Suppose that we perform a
change of reference frame to which corresponds a change of
coordinates given by  (rigid map)
\[
\begin{pmatrix} U & \bar{a} \\
\bar{0}^\intercal & 1 \end{pmatrix}, \qquad U\in SO(3),
\]
then in the new frame the original rigid motion gets represented by
the coordinate transformation matrix
\[
\begin{pmatrix} U & \bar{a} \\
\bar{0}^\intercal & 1 \end{pmatrix} \begin{pmatrix} O & \bar{b} \\
\bar{0}^\intercal
& 1 \end{pmatrix} \begin{pmatrix} U & \bar{a} \\
\bar{0}^\intercal & 1 \end{pmatrix}^{-1} .
\]
Chasles' theorem states that the new reference frame can be chosen
in such a way that the rigid motion in the newly defined coordinates
is
%
%
% i.e. to pass to
%coordinates
%\[
%\begin{pmatrix} \bar{x} \\ 1 \end{pmatrix} = \begin{pmatrix} \bar{y} \\
%1 \end{pmatrix},  \qquad \begin{pmatrix} \bar{x}' \\ 1 \end{pmatrix} =\begin{pmatrix} U & \bar{a} \\ \bar{0}^\intercal & 1 \end{pmatrix} \begin{pmatrix} \bar{y}' \\ 1 \end{pmatrix},
%\]
%where $U$ is a special orthogonal matrix, so that the rigid motion
%gets represented by the matrix
%\[
%\begin{pmatrix} U & \bar{a} \\
%\bar{0}^\intercal & 1 \end{pmatrix} \begin{pmatrix} O & \bar{b} \\ \bar{0}^\intercal
%& 1 \end{pmatrix} \begin{pmatrix} U & \bar{a} \\
%\bar{0}^\intercal & 1 \end{pmatrix}^{-1}
%\]
%in such a way that the rigid motion in the newly redefined
%coordinates becomes
\[{\small
\begin{pmatrix} x_1 '\\ x_2' \\x_3' \\ 1 \end{pmatrix} =\begin{pmatrix} 1 & 0 & 0 & -b \\ 0 & \cos \theta & \sin \theta & 0 \\
0 & -\sin \theta & \cos \theta & 0 \\
 0 & 0 & 0 & 1 \end{pmatrix} \begin{pmatrix} x_1 \\ x_2 \\ x_3 \\ 1
 \end{pmatrix},}
\]
where $b\in \mathbb{R}$ and $\theta\in [0,2\pi)$.  In other words,
the motion is a rotation about the first axis followed by a
translation of $b$ along the same axis (since these two operations
commute their order is irrelevant). If $\theta\ne 0$, the constant
$p\in \mathbb{R}$ such that $b=\frac{p}{2\pi}\, \theta$, is called
{\em pitch}.

It should be noted that a $\pi$-rotation of the reference frame on
the plane $\textrm{Span}(e_1,e_2)$ changes the sign of both $\theta$
and $b$. This operation makes it possible to choose the sign of $b$
or to impose $\theta \in [0,\pi]$. We shall impose $b>0$ whenever
$b\ne 0$, and $\theta \in [0,\pi]$ whenever $b=0$.

Algebrically, Chasles' theorem states that every conjugacy class in
the matrix group of maps given by Eq. (\ref{nug})  has a
representative of the above simplified form.

Here we are also interested in the reconstruction of the original
rigid motion starting from the conjugacy class and other geometric
data. The key observation is that by suitably limiting their
domains, the parameters $\theta$ and $b$ can be uniquely determined,
as they turn out to be independent of the simplifying reference
frame. In the same way, if $\theta\ne 0$ the first axis of the
simplifying reference frame does not depend on the frame (this is
the characteristic {\em axis of rotation}). Thus each rigid map
determines invariants of algebraic type (conjugacy class) and of
geometrical type. Once put together they allow us to fully recover
the rigid motion. Table \ref{table2} summarizes the families of
conjugacy classes, the relevant parameters and their domain, the
interpretation, and the necessary geometric ingredients needed to
recover the isometry given the conjugacy class (parameters).

For instance, line (c3) clarifies that we cannot recover the rigid
motion from the only information that it is a translation (i.e. type
(c3)) of module $b>0$. We need an additional normalized vector
belonging to $V$ which defines the direction of the translation
(indeed, in case (c3) the simplifying reference frame can be freely
translated, thus there is no characteristic line but only a
characteristic oriented direction). Similarly, if we know that the
isometry is a composition of a rotation and a translation
($\theta\in (0,2\pi), b>0$), then we need an oriented line in order
to recover the rigid motion. In the special case of a rotation of
angle $\pi$ with $b=0$, the orientation of the line is not needed
(indeed, the first axis of the simplifying reference frame can point
in both directions of the line).

We shall not comment these characterizations further as similar
considerations will be made for the relativistic generalization. We
end the section commenting  table \ref{table1} in which we lists the
conjugacy classes and the characteristic geometric invariants needed
to reconstruct the isometry in  Euler's case. It is worth noting
that if the direction and verse of the rotation are represented
using a normalized vector $v\in V$ then, joining the angle $\theta$
and this geometric object into $\theta v \in V$, we can represent
the  Lie group with a ball of radius $\pi$, in which opposite points
in the exterior spherical surface have been identified. This is a
well know geometrical representation of the group of rotations. This
construction shows that the conjugacy classes correspond to the
spherical surfaces inside the ball, the origin (the conjugacy class
of the identity), and the real projective plane of its surface (the
conjugacy class of $\pi$-rotations).

\subsection{Infinitesimal (Lie algebra) formulation and screw
product} \label{kwi}

The rigid motions $\psi:E \to E$ form a Lie group $R$. Near the
identity the exponential map is a diffeomorphism, thus there is some
element $v$ of the Lie algebra $\mathfrak{R}$ such that $\psi=\exp
(v s)$ for $s=1$. Every point $p\in E$ gives an orbit $s\to
\exp(vs)(p)$ and hence determines a vector tangent at $p$ which we
denote $v(p)$. Conversely, $v(p)$ determines a one-parameter group
of diffeomorphisms $\psi_s:E\to E$, and  $\psi=\psi_1$. Thus we may
identify the Lie algebra element $v$ with the vector field (denoted
in the same way) $v:E\to V$.

Suppose we have chosen a reference frame. The matrix  transformation
 $\psi=\exp (v
\epsilon)$ for small $\epsilon$ induces the coordinate change
\begin{equation} \label{nul}
\begin{pmatrix} \bar{x}' \\ 1
\end{pmatrix} =[I+\epsilon \begin{pmatrix} \Omega & \bar{c} \\
\bar{0}^\intercal & 0 \end{pmatrix}]
\begin{pmatrix} \bar{x} \\ 1 \end{pmatrix}
\end{equation}
where $\Omega\in \mathfrak{so}(3)$, i.e. it is a antisymmetric
matrix, while $\bar{c}$ is a 3-vector. Thus we can also identify
$\mathfrak{R}$ with $\mathfrak{iso}(3)$, namely the space of
matrices
of the form ${\footnotesize \begin{pmatrix} \Omega & \bar{c} \\
\bar{0}^\intercal & 0 \end{pmatrix}}$. This Lie algebra isomorphism
depends on the reference frame, as the matrix representing the
infinitesimal transformation changes under the Ad map of $ISO(3)$ on
$\mathfrak{iso}(3)$ for changes of frame.

Let us find the corresponding Lie algebra vector field. Let us
consider a point $q$ of coordinates $\bar{x}$ on the given starting
frame $\{\vec{e}_i\}$. This point is sent to $\psi(q)$, where
$\psi(q)$ is the point with the same coordinates $\bar{x}$ but in
the image frame $(o, \{\vec{e}_i\}) \to (\psi(o),
\{\psi_{*}(\vec{e}_i)\})$. This means that for the starting frame
$\psi(q)$ has coordinates $\bar{y}=\bar{x}-(\Omega
\bar{x}+\bar{c})\epsilon$. Thus the vector field $v:E\to V$  is
\[
v=-(\Omega_{ij} x^j +c^i)\vec{e}_i.
\]
We observe that a vector field satisfies the above equation for some
$\Omega\in \mathfrak{so}(3)$ and $\bar{c}$, if and only if there is
a vector $\vec\omega\in V$ such that for every $p,q\in E$
\begin{equation} \label{mfz}
v(p)-v(q)=\vec{\omega}\times (p-q).
\end{equation}
The previous equation is the {\em  constitutive equation of screws}
where a screw is nothing but a Lie algebra element of the group of
rigid motions. It can be shown that if a vector field is a screw
then $\vec\omega$ is uniquely determined. We call it the {\em screw
resultant}. If $\vec{\omega}\ne\vec{0}$ there is also a
characteristic line on $E$ called {\em screw axis}, which is the
locus at which $\vert v(p)\vert$ attains the minimum
\cite{minguzzi12,selig05}.

As we mentioned, the orbits of the Ad-action on $\mathfrak{iso}(3)$
might admit particularly simple representatives. This action
corresponds to frame changes, thus the choice of matrix
representative corresponds to a convenient frame choice. In
particular, we can obtain a simple representative choosing a frame
with the origin on the screw axis and first base element $\vec{e}_1$
aligned with the axis. In this way it is easy to show that the
representative takes the forms (Lc1) and (Lc2)  given by table
\ref{table2b}, respectively in case $\vec{\omega}=\vec{0}$ and in
case $\vec{\omega}\ne\vec{0}$.

On the Lie algebra of the group of rigid motions it is possible to
define an important Ad-invariant indefinite inner product called
{\em screw product}. Given two screws $v_1,v_2:E\to V$ we define
\begin{equation}
\langle v_1,v_2\rangle:= v_1(p)\cdot
\vec{\omega}_2+\vec{\omega}_1\cdot v_2(p).
\end{equation}
By using equation (\ref{mfz})  it can be easily shown that the
definition is well posed as the right-hand side is independent of
$p$. The screw product is particularly important in rigid body
dynamics were the kinetic energy and the power action on a rigid
body can be expressed through it \cite{minguzzi12,selig05}. Contrary
to a possible naive expectation, the screw product differs from the
Killing form of the Lie algebra \cite{minguzzi12} (which is instead
proportional to $\vec{\omega}_1\cdot \vec{\omega}_2$, namely the
scalar product of the resultants).

In a given reference frame the screw is determined by the pair
($\Omega$, $\bar{c}$). A calculation at the origin of the reference
frame shows that the screw product is given by
\begin{equation} \label{bsp}
\langle v_1,v_2\rangle=\frac{1}{2}[\epsilon_{ijk}\Omega^{(1)}_{ij}
c^{(2)}_k+\epsilon_{ijk}\Omega^{(2)}_{ij} c^{(1)}_k].
\end{equation}
It is clear that this expression is invariant under rotations of the
frame. Under translations the $\Omega$-s are left invariant while
the $\bar{c}$ terms change as follows $ c^{(2)}_k\to
-\Omega^{(2)}_{kj} b_j+c^{(2)}_k$, $ c^{(1)}_k\to -\Omega^{(1)}_{kj}
b_j+c^{(1)}_k$. The additional terms cancel out, hence the screw
product is Ad-invariant (for a different proof see
\cite{minguzzi12}). In section \ref{bos} we shall meet a kind of
relativistic generalization of the invariant (\ref{bsp}).

%Equation  is particularly interesting as it might give some hints on
%how to construct a relativistic screw product starting from the
%matrix representation of infinitesimal inhomogeneous Lorentz
%transformations.

\begin{table}[ht]
\caption{Euler's theorem  and reconstruction (Lie algebra version)}
{\small
\begin{tabular}{lcccc}
\\
\hline
Type & \parbox{3cm}{\begin{center}Families of \\ orbits \end{center}} & Parameters & Description & \parbox{1.7cm}{Geometric ingredients} \\
\hline
\\
(Le1) & {\footnotesize $\begin{pmatrix} 0 & 0 & 0  \\
0 &  0 & \theta  \\
0 & - \theta & 0
  \end{pmatrix}$} & $\theta \ne 0$ & \parbox{1.7 cm}{rotation field} & \parbox{1.7 cm}{direction and   verse} \\
  \\
 \hline
 \\
\end{tabular} }
\label{table1b}

\caption{Euler's theorem  and reconstruction (Group version)}
{\small
\begin{tabular}{lcccc}
\\
\hline
Type & \parbox{4cm}{\begin{center}Families of \\ conjugacy classes \end{center}} & Parameters & Description & \parbox{1.7cm}{Geometric ingredients} \\
\hline
\\
(e1) & {\footnotesize $\begin{pmatrix} 1 & 0 & 0  \\
0 & \cos \theta & \sin \theta  \\
0 & -\sin \theta & \cos \theta
  \end{pmatrix}$} & $\theta\in (0,\pi)$ & rotation & \parbox{1.7 cm}{direction and   verse} \\
 \\
(e2) & {\footnotesize $\begin{pmatrix} 1 & 0 & 0  \\ 0 & -1 & 0  \\
0 & 0 & -1  \\
 \end{pmatrix}$} & [none] & $\pi$-rotation & \parbox{1.7cm}{direction}  \\
 \\
 \hline
 \\
\end{tabular} }
\label{table1}
%\end{table}
\end{table}

\begin{table}[ht]
%\begin{table}

\caption{Chasles' theorem and reconstruction (Lie algebra version)}
{\small
\begin{tabular}{lcccc}
\\
\hline
Type & \parbox{4cm}{\begin{center}Families of \\ orbits \end{center}} & Parameters & Description & \parbox{1.7cm}{Geometric ingredients} \\
\hline
 \\
(Lc1) & {\footnotesize$\begin{pmatrix} 0 & 0 & 0 & -b \\ 0 & 0 & 0 & 0 \\
0 & 0 & 0 & 0 \\
 0 & 0 & 0 & 0 \end{pmatrix}$} & $b > 0$ &  \parbox{1.7 cm}{translation field} & \parbox{1.7 cm}{direction and   verse} \\
 \\
(Lc2) & {\footnotesize$\begin{pmatrix} 0 & 0 & 0 & -b \\ 0 & 0 &  \theta & 0 \\
0 & - \theta & 0 & 0 \\
 0 & 0 & 0 & 0 \end{pmatrix}$} & $\theta\ne 0$ & \parbox{1.7 cm}{screw field} & \parbox{1cm}{oriented \\ line} \\
 \\
 \hline\\
\end{tabular} }
\label{table2b}

\caption{Chasles' theorem and reconstruction (Group version)}
{\small
\begin{tabular}{lcccc}
\\
\hline
Type & \parbox{4cm}{\begin{center}Families of \\ conjugacy classes \end{center}} & Parameters & Description & \parbox{1.7cm}{Geometric ingredients} \\
\hline
\\
(c1) & {\footnotesize$\begin{pmatrix} 1 & 0 & 0 & 0 \\ 0 & \cos \theta & \sin \theta & 0 \\
0 & -\sin \theta & \cos \theta & 0 \\
 0 & 0 & 0 & 1 \end{pmatrix}$} & $\theta\in (0,\pi)$ & rotation & \parbox{1cm}{oriented \\ line} \\
 \\
(c2) & {\footnotesize$\begin{pmatrix} 1 & 0 & 0 & 0 \\ 0 & -1 & 0 & 0 \\
0 & 0 & -1 & 0 \\
 0 & 0 & 0 & 1 \end{pmatrix}$} & [none]  & $\pi$-rotation & \parbox{1cm}{line}  \\
 \\
(c3) & {\footnotesize$\begin{pmatrix} 1 & 0 & 0 & -b \\ 0 & 1 & 0 & 0 \\
0 & 0 & 1 & 0 \\
 0 & 0 & 0 & 1 \end{pmatrix}$} & $b > 0$ & translation & \parbox{1.7 cm}{direction and   verse} \\
 \\
(c4) & {\footnotesize$\begin{pmatrix} 1 & 0 & 0 & -b \\ 0 & \cos \theta & \sin \theta & 0 \\
0 & -\sin \theta & \cos \theta & 0 \\
 0 & 0 & 0 & 1 \end{pmatrix}$} & \parbox{1.5cm}{$b>0, \\ \theta\in (0,2\pi)$} & screw & \parbox{1cm}{oriented \\ line} \\
 \\
 \hline
\end{tabular} }
\label{table2}
\end{table}

\clearpage

\section{The Lorentz group} \label{tre}

Let $M$ be Minkowski spacetime, namely an affine space modeled over
the vector space $W$, where $(W,\eta, or, \uparrow)$ is a
4-dimensional vector space endowed with an inner product $\eta$ of
signature $(-,+,+,+)$, an orientation $or$, and a time orientation
$\uparrow$ (namely a choice of future and hence past timelike cone).
A vector in the future cone will be called {\em future directed},
f.d.\ for short. The {\em proper orthochronous Lorentz group}
$L^{\uparrow}_{+}$ is given by the set of automorphisms of $W$ which
respect both the orientation and the time orientation (an
automorphism  respects the time orientation if it sends the future
timelike cone into itself). The inhomogeneous proper orthochronous
Lorentz group $IL^{\uparrow}_{+}$ is made by the maps $P: M\to M$,
which preserve the inner product $\eta$, the orientation, and which
respect the time orientation. It can be shown (this fact can also be
deduced from Alexandrov and Zeeman's theorem
\cite{alexandrov50,alexandrov67,alexandrov67,zeeman64} on causal
automorphisms) that they are affine maps, namely they satisfy
$P(p+w)=P(p)+\Lambda(w)$, for every $p\in M$, $w\in W$, where
$\Lambda \in L^{\uparrow}_{+}$.

A {\em proper orthochronous orthonormal base} for $W$, is a
positively oriented tetrad $\{e_a; a=0,1,2,3\}$, such that $e_0$ is
timelike future directed and $\eta (e_a, e_b)=\eta_{ab}$ where
$\eta_{00}=-1$, $\eta_{ii}=1$, $i=1,2,3$, and the other values
vanish. Sometimes we shall refer to these bases as {\em reference
frames}. Once a reference frame has been chosen, any vector $w\in W$
can be written as $w=w^a e_a$ for some components $w^a\in
\mathbb{R}$, $a=0,1,2,3$.

Let $e_a'=\Lambda(e_a)$, then $\{e_a'\}$ is also a frame which can
be expressed in terms of the old base as $e_a'=(\Lambda^{-1})^b_{\
a} e_b$.  The change of reference frame induces a change in the
components of a vector $w\in W$ as follows  ${w^a}'=\Lambda^a_{\ b}
w^b$.
%If instead one writes down the coordinates of $\Lambda(w)$, (which
%is generically different  from $w$) in terms of the old base one
%obtains,  $w^b\to (\Lambda^{-1})^b_{\ a} w^a$ (active interpretation
%of the Lorentz transformation).
The choice of proper orthochronous orthonormal base  establishes an
isomorphism between the  Lorentz group $L^{\uparrow}_{+}$ and the
matrix proper orthochronous Lorentz group $SO(1,3)^{\uparrow}$ given
by the $4\times 4$ matrices $\Lambda^a_{\ b}$ such that
$\eta_{cd}=\eta_{ab} \Lambda^a_{\ c} \Lambda^b_{\ d}$,
$\textrm{det}(\Lambda^a_{\ b})=1$ and $\Lambda^0_{\ 0}>0$. Let us
focus on the action of $\Lambda$ on a different frame
$\tilde{e}_d=(L^{-1})^c_{\ d} e_c$. Let
$\tilde{e}_d'=\Lambda(\tilde{e}_d)=({\tilde{\Lambda}}^{-1})^c_{\ d}
\tilde{e}_c$, then $\tilde{\Lambda}^c_{\ d}=L^{c}_{\ a} \Lambda^a_{
\ b} (L^{-1})^{b}_{\ d}$. Thus, a change of frame acts as an
automorphism $g\to c g c^{-1}$ of $SO(1,3)^{\uparrow}$.

\subsection{The Lie algebra and its orbits} \label{igk}

The Lie algebra of the proper orthochronous Lorentz group
$L^{\uparrow}_+$ is given by the skew-symmetric linear maps $F: W\to
W$, that is by those maps such that, for every $w,v\in W$,
$\eta(v,Fw)+\eta(Fv,w)=0$. Any reference frame establishes a Lie
algebra isomorphism between this Lie algebra and the Lie algebra
$\mathfrak{so}(1,3)$ of the matrix group $SO(1,3)^{\uparrow}$
($\mathfrak{so}(1,3)^\uparrow$ and $\mathfrak{so}(1,3)$ coincide
because $SO(1,3)^{\uparrow}$ is the connected component of $O(1,3)$
which contains the identity)). As it is well known, $F^{a}_{\ b}\in
\mathfrak{so}(1,3)$ iff it is antisymmetric, $F_{ab}+F_{ba}=0$,
where the indices are lowered using $\eta_{cd}$.

The Ad-action of $SO(1,3)^\uparrow$ on $\mathfrak{so}(1,3)$ is given
by $F\to L F L^{-1}$. When $L$ runs over $SO(1,3)^\uparrow$ we get
an orbit of the Ad action on the Lie algebra. Each conjugacy
transformation represents a change of frame, thus by looking at a
convenient representative in the orbit we are looking at  frame
which simplifies the matrix expression of the infinitesimal
transformation.

%We mentioned that  is given by the vector space of the antisymmetric
%matrices endowed with the matrix commutator, namely by those
%matrices $F^{a}_{\ b}$ which satisfy
%This Lie algebra
%coincides with the Lie algebra of the group $O(1,3)$ (because
%$SO(1,3)^{\uparrow}$ is the connected component of $O(1,3)$ which
%contains the identity).
% Of course, any proper orthochronous
%orthonormal base establishes a bijection between this Lie algebra
%and that given by the skew-symmetric matrices.

The next result has long been established especially in connection
with electromagnetism (where  $F$ represents an electromagnetic
field). It can be regarded as a relativistic infinitesimal (i.e. Lie
algebra) version of Euler's theorem.

\begin{theorem} \label{vkx}
Let $F:W \to W$ be a skew-symmetric linear map, then it is possible
to choose a proper orthochronous orthornormal base $\{e_a\}$ such
that the endomorphism $F$ takes one of the following matrix forms
\begin{align*}
(a)  \quad A=
{\footnotesize \begin{pmatrix}
 0 & -\varphi  & 0 & 0 \\
 -\varphi & 0 & 0 & 0 \\
 0 & 0 & 0 & \theta \\
 0 & 0 & -\theta & 0
\end{pmatrix}};\qquad (b)
\quad B=
{\footnotesize\begin{pmatrix}
 0 & 0 & -\alpha & 0 \\
 0 & 0 & -\alpha & 0 \\
 -\alpha & \alpha & 0 & 0 \\
 0 & 0 & 0 & 0
\end{pmatrix}},
\end{align*}
where $\varphi > 0$, $\theta \in \mathbb{R}$, or $\varphi= 0$,
$\theta \ge 0$ and where $\alpha\in \mathbb{R}$ can be chosen at
will provided $\alpha \ne 0$. Stated in another way, the orbits of
$\mathfrak{so}(1,3)$ under the Ad action of $SO(1,3)^\uparrow$ admit
one and only one of the representatives given above (apart for the
mentioned freedom in $\alpha$) (the trivial orbit of the origin
contains only the zero matrix).

Defined the invariants
\begin{align*}
I_1&=\frac{1}{4}F_{a b} F^{a b}=-\frac{1}{4} \textrm{Tr} F^2,\\
I_2&= -\frac{1}{4}\epsilon_{a b c d} F^{a b} F^{c d},
\end{align*}
where $\epsilon_{0123}=1$, we have $I_1=(\theta^2-\varphi^2)/2$,
$I_2=\theta \varphi$, thus it is possible to read the orbit
calculating
\begin{align}
\varphi&=\sqrt{-I_1+\sqrt{I_1^2+I_2^2}}, \label{vt1}\\
\theta&=\textrm{sgn}(I_2) \sqrt{I_1+\sqrt{I_1^2+I_2^2}}, \label{vt2}
\end{align}
(where $\textrm{sgn}(0)=1$) provided $\varphi$ or $\theta$ is
different from zero (i.e. if we happen to be in case (a) where at
least one of the invariant does not vanish). The map $F$ is
non-singular if and only if $I_2\ne 0$.
 \end{theorem}

\begin{proof}
A proof of the first claim can be found in \cite[Sect.
2.4]{naber92}, \cite[Sect. 9.5]{synge56}, \cite[Sect.
9.3]{stephani03} or \cite{misner73}. The latter claims follow
easily. We give here a simple proof of the first claim. We start
choosing any base. The matrix of the endomorphism $F$ takes the form
${\footnotesize \begin{pmatrix}
 0 & c_1  & c_2 & c_3 \\
 c_1 & 0 & b_3 & -b_2 \\
 c_2 & -b_3 & 0 & b_1 \\
 c_3 & b_2 & -b_1 & 0
\end{pmatrix}}$. Under rotations of the reference frame the triples
$\vec{c}=(c_1,c_2,c_3)$ and $\vec{b}=(b_1,b_2,b_3)$ transform as
vectors. The invariants read $I_1=\frac{1}{2}(\vec{b}^2-\vec{c}^2)$,
$I_2=-\vec{c}\cdot \vec{b}$. We can choose the frame in such a way
that $\vec{c}\propto {e}_2$, $\vec{c},\vec{b}\in
\textrm{Span}(e_2,e_3)$, and the first axis is oriented in the
direction of $\vec{c} \times \vec{b}$ (so that $\vec{c} \times
\vec{b}=(c_2 b_3,0,0)$ where $c_2 b_3\ge 0$). This choice simplifies
the matrix because $c_1=c_3=b_1=0$. Furthermore, if the invariants
vanish then $\vec{c}$ and $\vec{b}$ are perpendicular and of the
same magnitude, thus we can choose $\vec{b}$ aligned with the third
axis, and hence obtain (b). If $\vec{b}\propto \vec{c}$  then we
obtain (a) aligning $e_1$ with them. In the remaining case
$c_2b_3=\vert \vec{c}\times \vec{b}\vert>0$. Now, we make a boost in
direction $e_1$ so that the endomorphism gets represented by the
matrix
\[
{\footnotesize \begin{pmatrix}
 \cosh \gamma & -\sinh \gamma  & 0 & 0  \\
 -\sinh \gamma & \cosh \gamma & 0 & 0  \\
 0 & 0 & 1 &  0  \\
 0 & 0 & 0 & 1
\end{pmatrix}  \begin{pmatrix}
 0 & 0  & c_2 & 0 \\
 0 & 0 & b_3 & -b_2 \\
 c_2 & -b_3 & 0 & 0 \\
 0 & b_2 & 0 & 0
\end{pmatrix}\begin{pmatrix}
 \cosh \gamma & \sinh \gamma  & 0 & 0  \\
 \sinh \gamma & \cosh \gamma & 0 & 0  \\
 0 & 0 & 1 &  0  \\
 0 & 0 & 0 & 1
\end{pmatrix} },
\]
so that $c_2'=c_2\cosh\gamma-b_3\sinh\gamma$,
$b_3'=-c_2\sinh\gamma+b_3\cosh\gamma$, $c_3'=b_2\sinh\gamma$,
$b_2'=b_2\cosh\gamma$ and $c_1'=b_1'=0$. We ask if we can find a
value of $\gamma$ which aligns $\vec{c}'$ with $\vec{b}'$. They are
aligned if $\vec{c}'\times \vec{b}'=\vec{0}$ which holds if the next
expression vanishes
\[
c_2'b_3'-c_3'b_2'=-c_2^2\sinh^2\gamma-(b_2^2+b_3^2)\sinh \gamma
\cosh\gamma+c_2b_3\cosh (2\gamma).
\]
For $\gamma=0$ the right-hand side gives $c_2b_3>0$, while for large
$\gamma$ it goes as $ \sim ( -\vec{c}^2-\vec{b}^2+2\vert
\vec{c}\times \vec{b}\vert) e^{2\gamma}/4$. Thus if
$\vec{c}^2+\vec{b}^2>2\vert \vec{c}\times \vec{b}\vert$ it vanishes
for some $\gamma$. This is the case because
\[
( \vec{c}^2+\vec{b}^2)^2-(2\vert \vec{c}\times
\vec{b}\vert)^2=(\vec{c}^2-\vec{b}^2)^2+4(\vec{c}\cdot\vec{b})^2=4(I_1^2+I_2^2)>0.
\]

%
%For a proof of the first claim see  \cite[Sect. 2.4]{naber92},
%\cite[Sect. 9.5]{synge56}, \cite[Sect. 9.3]{stephani03} or
%\cite{misner73}. The latter claims follow easily. With respect to
%the former reference we use the indices 0,1,2,3, instead of 1,2,3,4,
%where 0 plays  the role of `time' index. Furthermore, we have
%operated a cyclic permutation $1\to 2 \to 3 \to 1$ of the space
%indices.

\end{proof}

The following identity shows that, indeed, the orbit of $B(\alpha)$,
contains all the  matrices of type $B(\alpha')$, with $\alpha'\ne 0$
(in order to  change sign make a $\pi$-rotation of the frame on the
plane $\textrm{Span}(e_2,e_3)$)

\begin{align}
 &{\footnotesize \begin{pmatrix}
 \cosh \gamma & -\sinh \gamma  & 0 & 0  \\
 -\sinh \gamma & \cosh \gamma & 0 & 0  \\
 0 & 0 & 1 &  0  \\
 0 & 0 & 0 & 1
\end{pmatrix}
\begin{pmatrix}
 0 &  0 & -\alpha & 0 \\
 0 & 0 & -\alpha  & 0 \\
 -\alpha & \alpha & 0 & 0 \\
 0 & 0 & 0 & 0
\end{pmatrix}
\begin{pmatrix}
 \cosh \gamma & \sinh \gamma  & 0 & 0  \\
 \sinh \gamma & \cosh \gamma & 0 & 0  \\
 0 & 0 & 1 &  0  \\
 0 & 0 & 0 & 1
\end{pmatrix} }\nonumber \\
&= {\footnotesize \begin{pmatrix}
 0 &  0 & -{\alpha'}  & 0 \\
 0 & 0 & -{\alpha'}   & 0 \\
 -{\alpha'}  & {\alpha'}  & 0 & 0 \\
 0 & 0 & 0 & 0
\end{pmatrix}}, \qquad \text{with} \ \alpha'=\alpha e^{-\gamma} .
\label{mkd}
\end{align}

\begin{remark}
The orbit with representative $B$ cannot be distinguished from the
trivial orbit using continuous invariant functions. Indeed, if $[B]$
is the orbit of $B$, then $\overline{[B]}$ contains the identity
(take the limit $\alpha\to 0$ of $B(\alpha)$) and hence the function
would take the same value on both orbits.
\end{remark}

\begin{remark}
The physical content of the previous theorem is quite interesting.
It tells us that any infinitesimal Lorentz transformation can be
regarded as the frame dragging of points attached to a frame which
is accelerating and rotating in two canonical ways. One with the
acceleration and angular velocities aligned (we can choose the first
axis with a suitable rotation), and the other with  acceleration and
angular velocities which are perpendicular and of equal module. It
also tells us that in the former case the modules of the
acceleration and of the angular velocity do not depend on the frame
that accomplishes the simplification, and thus can be regarded as
genuine characteristics of the infinitesimal Lorentz transformation.
In the latter case, on the contrary, the (equal) modules are not
uniquely determined because they depend on the simplifying frame.
Indeed, they change boosting the simplifying frame along the
direction determined by the vector product between acceleration and
angular velocity.

While this interpretation is correct, it should be kept in mind that
Lorentz transformations act on $W$, not on $M$. As we shall see, the
introduction of translations will allow us to assign, for fixed
movement duration, a meaningful module to the acceleration and
angular velocity, even in those cases in which they are
perpendicular.

\end{remark}

%It is convenient to record the following fo
\begin{remark} \label{nur}
We have the following identities
\[
e^A= {\footnotesize\begin{pmatrix}
 \cosh \varphi  & -\sinh \varphi   & 0 & 0 \\
 -\sinh \varphi  & \cosh \varphi  & 0 & 0 \\
 0 & 0 & \cos \theta  &  \sin \theta  \\
 0 & 0 & -\sin \theta  & \cos \theta
\end{pmatrix}},
\quad e^B= {\footnotesize\begin{pmatrix}
 1+ \alpha^2/2 & -\alpha^2/2  &  -\alpha  & 0 \\
\alpha^2/2 & 1-\alpha^2/2 &  -\alpha  & 0 \\
  -\alpha & \alpha & 1 & 0 \\
 0 & 0 & 0 & 1
\end{pmatrix}}.
\]
%
%Now, let us consider the matrices
%\begin{align*}
%M(s)&=
%\begin{pmatrix}
% \cosh (\varphi s) & -\sinh (\varphi s)  & 0 & 0 \\
% -\sinh (\varphi s) & \cosh (\varphi s) & 0 & 0 \\
% 0 & 0 & \cos (\theta s) &  \sin (\theta s) \\
% 0 & 0 & -\sin (\theta s) & \cos (\theta s)
%\end{pmatrix}, \\
% N(s)&=
%\begin{pmatrix}
% 1+ (\alpha s)^2/2 & -(\alpha s)^2/2  &  -\alpha s  & 0 \\
%(\alpha s)^2/2 & 1-(\alpha s)^2/2 &  -\alpha s  & 0 \\
%  -\alpha s & \alpha s & 1 & 0 \\
% 0 & 0 & 0 & 1
%\end{pmatrix}.
%\end{align*}
%They belong to $SO(1,3)^{\uparrow}$; for $t=0$ they coincide with
%the identity, and it is easy to verify that $\frac{\dd }{\dd s}
%M(s)=A M(s)$ and $\frac{\dd }{\dd s} N(s)=B N(s)$. As a consequence,
%$M(s)=\exp (A s)$ and $N(s)=\exp (B s)$.
It is interesting to note that $\exp A$ preserves the null
directions $e_0\pm e_1$, while $\exp B$ leaves invariant the null
vector $e_0+e_1$.

\end{remark}

\begin{proposition} \label{nux}
Let $\Lambda \in SO(1,3)^\uparrow$ and $F\in \mathfrak{so}(1,3)$ be
such that $\Lambda = e^F$, then $\Lambda$ and $F$ have the same
eigenvectors.
\end{proposition}

\begin{proof}
Of course, it is trivial that the eigenvectors for $F$ are
eigenvectors for $\Lambda$. The non-trivial direction is the
opposite. It is easy to check that the claim holds for $F=A$ or
$F=B$. Since by a conjugacy transformation we can always reduce the
problem to this case, the claim holds in general.
\end{proof}

A Lie group $G$ with a surjective exponential map $\exp:
\mathfrak{g}\to G$ is called {\em exponential}.  The Lorentzian
generalization of Euler's and Chasles' theorems can
 be obtained from their infinitesimal versions thanks to the
 following result.
%
%
%  and their Lorentzian generalizations can be proved trying to
%simplify through conjugacy the matrix representation of a Lie
%algebra element rather than  a group matrix element.
%
%This approach works provided one proves that the Lie groups under
%study are {\em exponential}, namely that the exponential map from
%the Lie algebra to the Lie group is surjective. For the Lorentz
%group we have

\begin{theorem} \label{nco}
The exponential map $\exp: \mathfrak{so}(1,3) \to
SO(1,3)^{\uparrow}$ is surjective.
\end{theorem}

%\begin{proof}
Exponential Lie groups are very much studied in the literature
\cite{dokovic97} and the previous result is well established
\cite{riesz93,shaw69b,nishikawa83,dokovic97,moskowitz99}
\cite[Theor. 6.5]{hall04} \cite[Theor. 4.21]{gallier12}, see also
\cite{coll90,coll02}.
 The nice fact is that although
$SL(2,\mathbb{C})$ provides a double covering of
$SO(1,3)^{\uparrow}$, the exponential $\exp:
\mathfrak{sl}(2,\mathbb{C}) \to SL(2,\mathbb{C})$ is not surjective
\cite{moskowitz99,gallier12} (the group $SL(2,\mathbb{R})$ is often
used to shown  that the exponential does not need to be surjective
\cite{duistermaat00}). Indeed, the matrix
\[
\begin{pmatrix}
-1 & h\\ 0 & -1
\end{pmatrix}, \qquad h\ne 0
\]
does not belong to any 1-parameter subgroup of $SL(2,\mathbb{C})$.

\subsection{Lorentzian extension of Euler's theorem}

We formulate the Lorentzian generalization of Euler's theorem.

\begin{theorem} \label{nhp}
Let $\Lambda: W\to W$ be a non-trivial proper orthochronous Lorentz
transformation, then we can find a proper orthochronous orthonormal
base in such a way that  the matrix $\Lambda^a_{\ b}$ belongs to the
2-dimensional Abelian subgroup of roto--boosts
\begin{align*}
(a): \ {\footnotesize\begin{pmatrix}
 \cosh \varphi & -\sinh \varphi  & 0 & 0  \\
 -\sinh \varphi & \cosh \varphi & 0 & 0  \\
 0 & 0 & \cos \theta &  \sin \theta  \\
 0 & 0 & -\sin \theta & \cos \theta
\end{pmatrix}}, \qquad \textrm{with: } \ \parbox{3.5cm}{$\varphi > 0,\ \theta \in [0,2\pi)$, or \\ $\varphi=0,\ \theta \in [0,\pi]$}
\end{align*}
or  to the 1-dimensional Abelian subgroup of null (Galileian) boosts
\begin{align*}
(b): \ {\footnotesize\begin{pmatrix}
 1+ \alpha ^2/2 &  -\alpha^2/2 & -\alpha & 0 \\
 \alpha^2/2 & 1-\alpha ^2/2 & -\alpha  & 0 \\
 -\alpha & \alpha & 1 & 0 \\
 0 & 0 & 0 & 1
\end{pmatrix}}, \qquad \qquad \alpha \in \mathbb{R}. \qquad\qquad \qquad\qquad \qquad
\end{align*}
%If in case
%(a) the angle $\theta$ can be chosen in $[0,\pi]$.
If (b) applies with $\alpha \ne 0$ then $\alpha$ can be chosen
arbitrarily (as long as it is different from zero). Apart from this
freedom, which does not change the conjugacy class, different
matrices correspond to different conjugacy classes.

The matrix is of type (a) if and only if $\Lambda:W\to W$ leaves
invariant a timelike 2-subspace, if and only if $\Lambda:W\to W$
leaves invariant at least two lightlike directions. The matrix is of
type (a) with $\theta=0$, if and only if it is of type (a) and
leaves invariant one spacelike vector (and hence every vector in a
spacelike 2-subspace). The matrix is of type (a) with $\varphi=0$,
if and only if it is of type (a) and leaves invariant at least two
lightlike vectors. The matrix is of type (b) with $\alpha\ne 0$ if
and only if $\Lambda:W\to W$ leaves invariant one and only one
lightlike vector.

More specifically, the reference frame can be chosen in such a way
that the matrix takes one and only one of the forms given in table
\ref{table3}. The type and the parameters' value are independent of
the simplifying reference frame and, moreover, the simplifying
reference frame fixes unambiguously some geometric data given in the
last column of the table. Furthermore, if the type, the parameters
and the geometric data are given, then the transformation can be
completely determined.

\end{theorem}

\begin{proof}
By theorem \ref{nco}  there is some antisymmetric matrix $F$ such
that $\Lambda= \exp F$. We choose the proper orthochronous
orthonormal base in such a way that $F$ takes one of the canonical
forms given by theorem \ref{vkx}. Thus by suitably choosing the
proper orthochronous orthonormal base we can make
$\Lambda$ to take the form  $\exp A$ or $\exp B$ given by remark
\ref{nur}. This is  the first claim of the  theorem. From here the
other statements follow with little effort.
\end{proof}

Transformations of type (b) might be called  {\em Galileian boosts}.
The justification of this terminology can be found in
\cite{minguzzi05e}, where it is shown that they provide a
1-dimensional subgroup of the group $E(2)$ of Galileian boost in
2-dimensions plus rotations (see also \cite{duval85,duval91}).

\begin{table}[ht]
\caption{Relativistic Euler's theorem  and reconstruction (Lie
algebra version)} {\small
\begin{tabular}{lllll}
\\
\hline
Type & \parbox{4cm}{\begin{center}Families of \\ orbits \end{center}} & Parameters & Description & \parbox{1.7cm}{Geometric ingredients} \\
\hline
\\
(Ll1) & {\footnotesize $\begin{pmatrix}
 0 & - \varphi  & 0 & 0  \\
 - \varphi & 0 & 0 & 0  \\
 0 & 0 & 0 &   \theta  \\
 0 & 0 & - \theta & 0
\end{pmatrix}$} & \parbox{2cm}{$\varphi=0$, $\theta > 0$, \\or  $\varphi>0$}
  &  \parbox{1.7 cm}{roto--boost field}  & \parbox{1.7 cm}{oriented timelike 2-subspace} \\
 \\
(Ll2) & {\footnotesize $\begin{pmatrix}
 0 &  0 & -1 & 0 \\
 0 & 0 & -1 & 0 \\
 -1 & 1 & 0 & 0 \\
 0 & 0 & 0 & 0
\end{pmatrix}$}
  &  [none]  & \parbox{1.5cm}{null (Galileian) boost field}
  & \parbox{1.8cm}{oriented lightlike 2-subspace \mbox{and f.d.\ } lightlike vector on it.}  \\
 \\
 \hline
\end{tabular} }
\label{table3b}
\end{table}

\begin{table}[ht]
\caption{Relativistic Euler's theorem  and reconstruction (Group
version) } {\small
\begin{tabular}{lllll}
\\
\hline
Type & \parbox{4cm}{\begin{center}Families of \\ conjugacy classes \end{center}} & Parameters & Description & \parbox{1.7cm}{Geometric ingredients} \\
\hline
%\\
%(l1) & {\footnotesize $\begin{pmatrix}
% \cosh \varphi & \sinh \varphi  & 0 & 0  \\
% \sinh \varphi & \cosh \varphi & 0 & 0  \\
% 0 & 0 & 1 &  0  \\
% 0 & 0 & 0 & 1
%\end{pmatrix}$} & $\varphi >0$
%  & boost & \parbox{1.7 cm}{oriented timelike 2-subspace} \\
\\
(l1) & {\footnotesize $\begin{pmatrix}
 1 & 0  & 0 & 0  \\
 0 & 1 & 0 & 0  \\
 0 & 0 & \cos \theta &  \sin \theta  \\
 0 & 0 & -\sin \theta & \cos \theta
\end{pmatrix}$} & $\theta\in (0,\pi)$
  & rotation & \parbox{1.7 cm}{oriented timelike 2-subspace} \\
  \\
(l2) & {\footnotesize $\begin{pmatrix}
 1 & 0  & 0 & 0  \\
 0 & 1 & 0 & 0  \\
 0 & 0 & -1 &  0  \\
 0 & 0 & 0 & -1
\end{pmatrix}$} & [none]
  & $\pi$-rotation & \parbox{1.7 cm}{timelike \mbox{2-subspace}} \\
\\
(l3) & {\footnotesize $\begin{pmatrix}
 \cosh \varphi & -\sinh \varphi  & 0 & 0  \\
 -\sinh \varphi & \cosh \varphi & 0 & 0  \\
 0 & 0 & \cos \theta &  \sin \theta  \\
 0 & 0 & -\sin \theta & \cos \theta
\end{pmatrix}$} & $\varphi >0,\, \theta\in [0,2\pi)$
  & roto--boost & \parbox{1.7 cm}{oriented timelike 2-subspace} \\
 \\
(l4) & {\footnotesize $\begin{pmatrix}
 1+ \alpha ^2/2 &  -\alpha^2/2 & -\alpha & 0 \\
 \alpha^2/2 & 1-\alpha ^2/2 & -\alpha  & 0 \\
 -\alpha & \alpha & 1 & 0 \\
 0 & 0 & 0 & 1
\end{pmatrix}$}
  &  \parbox{2cm}{Any $\alpha\ne 0$ gives the same conjugacy class}  & \parbox{1.5cm}{null (Galileian) boost}
  & \parbox{1.8cm}{oriented lightlike 2-subspace \mbox{and f.d.\ } lightlike vector on it.}  \\
 \\
 \hline
\end{tabular} }
\label{table3}
\end{table}

\clearpage

\begin{remark}
Let us clarify the role of the geometric data.

Suppose that (a) applies with $\varphi = 0$, $\theta\ne 0$. Choose a
simplifying frame such that $\theta\in [0,\pi]$. The invariant
timelike 2-dimensional subspace $\textrm{Span}(e_0,e_1)$ selected in
this way is independent of the simplifying reference frame (precisely because it is characterized as timelike invariant subspace of $\Lambda: W\to W$).
Furthermore, if $\theta\ne \pi$ the orientation of this subspace
given by $(e_0,e_1)$ is independent of the simplifying reference
frame.

Suppose that (a) applies with $\varphi \ne 0$. Choose a simplifying
frame such that $\varphi>0$, and assign to the invariant timelike
2-dimensional subspace $\textrm{Span}(e_0,e_1)$ the orientation
given by the base $(e_0, e_1)$. The invariant oriented timelike
2-subspace selected in this way is independent of the simplifying
reference frame.

Suppose that (b) applies with $\alpha \ne 0$. Choose a simplifying
frame such that $\alpha=1$, and assign to the invariant lightlike
2-subspace $\textrm{Span}(e_0+e_1,e_3)$ the orientation given by the
base $(e_0+e_1,e_3)$. The invariant oriented lightlike 2-subspace
selected in this way and the lightlike vector $e_0+e_1$ are
independent of the simplifying reference frame (for, another simplifying frame would be related to the former by a little group transformation of the vector $(1,1,0,0)$. From here, since a null 2-plane must be left invariant, the frame change matrix must actually be a Galileian boost \cite{minguzzi05e} with direction in $\textrm{Span}(e_2,e_3)$.).

The map $\Lambda$ can be completely recovered knowing, to start
with, if it is   of type (a) or (b). In case (a)  it is
sufficient to know the invariant oriented timelike 2-subspace and
the constants $\vert \varphi \vert$, $\theta $.  In case (b) it is
sufficient to know the invariant oriented lightlike
2-subspace which admits a base of invariant vectors, and
the distinguished future directed lightlike vector on it.

%\footnote{We are not claiming that there is only one
%invariant lightlike 2-subspace. Rather we refer to the fact that
%only one can be the result of the procedure mentioned above.}
\end{remark}

The many paragraphs of the theorem serve to clarify the qualitative
features of the Lorentz transformations. Two transformation which
differ by these aspects cannot be related by conjugacy (for other
characterizations see \cite{shaw69} \cite[Theor. 6.1]{hall04}).

We mention here another interesting approach to the study of
conjugacy classes. It uses the isomorphism between the Lorentz group
and $PSL(2,\mathbb{C})$. The idea comes from the observations that
Lorentz transformations of the observer induce an action of
$PSL(2,\mathbb{C})$ on the Riemann sphere of the 'night sky'
\cite{penrose59,penrose84,naber92}.

%The conjugacy classes for $PSL(2,\mathbb{C})$ are well known
%\cite{needham97}.

According to the classification of conjugacy classes for the
$PSL(2,\mathbb{C})$ group \cite{needham97}, the classes
corresponding to matrices of the form (a) with $\theta=0$,
$\varphi\ne 0$, are called {\em hyperbolic}, those corresponding to
matrices of the form (a) with $\theta\ne 0$, $\varphi= 0$, are
called {\em elliptic}, those corresponding to matrices of the form
(a) with $\theta\ne 0$, $\varphi\ne  0$, are called {\em
loxodromic}, and that corresponding to matrices of the form (b) with
$\alpha \ne 0$, is called {\em parabolic}. The class of the identity
contains only the identity and is referred as the trivial class. We
shall extend this terminology to the Lorentz transformations
themselves. Thus a Lorentz transformation is hyperbolic if its
conjugacy class is hyperbolic.

Since the parameters $\varphi$ and $\theta$ are expressible in terms of Ad-invariants, matrices obtained
for distinct parameters correspond to distinct conjugacy classes (up
to the remarked freedom in the parameters  choice).

%Of course, the characterizations given in the theorem do not show
%that inside the same family, say hyperbolic, the matrices obtained
%for distinct parameters correspond to distinct conjugacy classes (up
%to the remarked freedom in the parameters  choice). However,  it is
%not difficult to prove it once the spaces left invariant by these
%actions are known.

Let us instead show that the parameters have the mentioned freedom.
Suppose we are in case (a) with $\varphi=0$. A $\pi$-rotation of the
reference frame in the plane $\textrm{Span}(e_1, e_2)$, changes the
sign of $\theta$ which is then redefined adding $2\pi$. As a result
$\theta$ can be changed to take value in $[0,\pi]$ .

Suppose we are in case (a) with $\varphi\ne 0$.  In order to show
that only the sign of $\varphi$ is relevant for the conjugacy class,
it is again sufficient to perform a $\pi$-rotation of the reference
frame in the plane $\textrm{Span}(e_1, e_2)$, as it changes the
signs of both $\varphi$ and $\theta$ which is then redefined adding
$2\pi$.

%We remark that in case (a) the timelike 2-subspace  has to be
%oriented in such a way that,  choosing $e_0$ on the plane and future
%directed, and $e_1$ on the plane so as to form a positive oriented
%base for the plane, the matrix $\Lambda^a_{\ b}$ becomes such that
%$\varphi\ge 0$ (clearly, changing sign to $e_1$ changes sign to
%$\varphi$).

In case (b) different modules for $\alpha$ do not give different
conjugacy classes  (if $\alpha=0$ we get the class of the identity),
because of the identity

{
\begin{align}
 &{\footnotesize \begin{pmatrix}
 \cosh \gamma & -\sinh \gamma  & 0 & 0  \\
 -\sinh \gamma & \cosh \gamma & 0 & 0  \\
 0 & 0 & 1 &  0  \\
 0 & 0 & 0 & 1
\end{pmatrix}
\begin{pmatrix}
 1+ \alpha ^2/2 &  -\alpha^2/2 & -\alpha & 0 \\
 \alpha^2/2 & 1-\alpha ^2/2 & -\alpha  & 0 \\
 -\alpha & \alpha & 1 & 0 \\
 0 & 0 & 0 & 1
\end{pmatrix}
\begin{pmatrix}
 \cosh \gamma & \sinh \gamma  & 0 & 0  \\
 \sinh \gamma & \cosh \gamma & 0 & 0  \\
 0 & 0 & 1 &  0  \\
 0 & 0 & 0 & 1
\end{pmatrix} }\nonumber \\
&= {\footnotesize \begin{pmatrix}
 1+ {\alpha'} ^2/2 &  -{\alpha'} ^2/2 & -{\alpha'}  & 0 \\
 {\alpha'} ^2/2 & 1-{\alpha'} ^2/2 & -{\alpha'}   & 0 \\
 -{\alpha'}  & {\alpha'}  & 1 & 0 \\
 0 & 0 & 0 & 1
\end{pmatrix}}, \qquad \text{with} \ \alpha'=\alpha e^{-\gamma} .
\label{mkc}
\end{align}
}

The sign of $\alpha$ is also irrelevant because  a $\pi$-rotation of
the reference frame on the plane $\textrm{Span}(e_2,e_3)$ replaces
$\alpha$ with $\alpha'=-\alpha$.
%
%we get {\small
%\begin{align*}
% &\begin{pmatrix}
% 1 & 0  & 0 & 0  \\
% 0 & 1 & 0 & 0  \\
% 0 & 0 & -1 &  0  \\
% 0 & 0 & 0 & -1
%\end{pmatrix}
%\begin{pmatrix}
% 1+ \alpha ^2/2 &  -\alpha^2/2 & \alpha & 0 \\
% \alpha^2/2 & 1-\alpha ^2/2 & \alpha  & 0 \\
% \alpha & -\alpha & 1 & 0 \\
% 0 & 0 & 0 & 1
%\end{pmatrix}
%\begin{pmatrix}
% 1 & 0  & 0 & 0  \\
% 0 & 1 & 0 & 0  \\
% 0 & 0 & -1 &  0  \\
% 0 & 0 & 0 & -1
%\end{pmatrix}\\
%&=
%\begin{pmatrix}
% 1+ {\alpha'} ^2/2 &  -{\alpha'} ^2/2 & {\alpha'}  & 0 \\
% {\alpha'} ^2/2 & 1-{\alpha'} ^2/2 & {\alpha'}   & 0 \\
% {\alpha'}  & -{\alpha'}  & 1 & 0 \\
% 0 & 0 & 0 & 1
%\end{pmatrix}, \qquad \text{with} \ \alpha'=-\alpha.
%\end{align*}
%}
Thus, there will be a reference frame for which $\alpha =1$. The
lightlike vector $e_0+e_1$ in this frame is then rather special
(what is special of it is of course its normalizing module), and it
is the invariant future directed lightlike vector to which the
reconstruction statement of the theorem refers to.

Let us justify the reconstruction claim of the theorem. Suppose, for
instance, that we are given the invariant timelike oriented
2-subspace and constants $\varphi> 0$, $\theta\in [0,\pi]$. We
choose a vector $e_0$, timelike, normalized and future directed on
the plane and $e_1$ orthogonal to it in such a way that $(e_0,e_1)$
is positively oriented. Next we choose $(e_2,e_3)$ spacelike,
normalized, and orthogonal among themselves and with respect to the
distinguished timelike plane. We order them in such a way that
$(e_0,e_1,e_2,e_3)$ is positively oriented. The linear  map
$\Lambda: W\to W$ whose matrix form in the base $\{e_a\}$ is given
by matrix (a) is well defined as it is independent of the chosen
base. The independence comes from the fact that the matrix form (a)
is invariant under boosts of the frame in the $e_1$ direction and
under rotations on the spacelike plane $\textrm{Span}(e_2,e_3)$.

The first part of theorem \ref{nhp} is essentially known
\cite{wigner39,abraham48,lomony63,schwartz63,schwartz65,shaw69}.
Sometimes the  conjugacy classes are incorrectly identified, either
confusing the family of classes as one single conjugacy class, or
not realizing that the matrices of type (b) with $\alpha\ne 0$
belong to the same conjugacy class (we stress that there is only one
parabolic conjugacy class).

The correct identification of the conjugacy classes is important but
this data does not allow us to  recover the transformation
$\Lambda$. Two transformations belong to the same conjugacy class if
it is possible to find two observers (proper orthochronous
orthonormal bases) with respect to which they look the same
\cite{shaw69}.

The last sentence of the theorem allows us to extract the true
physical content of the transformation $\Lambda$. It is important to
realize that in case (a) with $\varphi\ne 0$, the timelike oriented
2-subspace and the coefficients $\vert \varphi \vert$, $\theta$,
have physical significance as they are independent of the reference
frame. In the same way it is important to realize that in case (b),
$\alpha$ has no physical significance while the oriented lightlike
2-subspace and the normalizing lightlike vector on it, do have. Thus
 the same parabolic conjugacy class corresponds to different
parabolic Lorentz transformations $\Lambda$, as there are many
distinct pairs made by an oriented lightlike 2-subspace and a future
directed lightlike vector on it.

This analysis shows that  the rotation axis of Euler's theorem is
replaced here by an oriented causal plane passing through the origin
in which the future direction is suitably normalized (this
normalization can be omitted in the timelike case given the
existence of a Lorentzian induced metric). The rotation angle in
Euler's theorem is instead replaced by parameters $\vert
\varphi\vert$, $\theta$ (in case (a)).

%Given these observations one could ask why we did not restrict the
%parameter $\varphi$ to the domain $[0,+\infty)$ and the parameter
%$\alpha$ to $\{0,+1\}$ since the beginning. We preferred to define
%families (a) and (b) so as to form abelian {\em subgroups}. As we
%shall realize more clearly with Chasles' theorem, it is convenient
%to do that because in this way the interpretation of the
%representative transformations simplifies considerably.

Given a Lorentz transformation $\Lambda: W\to W$ it is possible to
read its conjugacy class through its characteristic polynomial
$p(\lambda)=\det(\Lambda -\lambda I)$.

\begin{theorem}
Let $\Lambda\in L^{\uparrow}_+$, $\Lambda\ne I$, then  the
characteristic polynomial reads
\begin{align*}
&(a')  \qquad p(\lambda)=
(\lambda^2-2\cos\theta \,\lambda+1)(\lambda^2-2\cosh\varphi\,\lambda+1),\\
&(b') \qquad p(\lambda)= (\lambda-1)^4,
\end{align*}
where (a') holds iff case (a)  of theorem \ref{nhp} applies, and
(b') holds iff case (b) of theorem \ref{nhp} applies. In particular,
it is possible to distinguish between cases (a) and (b) and, if case
(a) applies, to read $\varphi \in [0,+\infty)$, and $\theta \in
[0,2\pi)$. Indeed, let $p(\lambda)=\lambda^4-p_3
\lambda^3+p_2\lambda^2-p_1\lambda +1$, then
\begin{align*}
p_3&=p_1=2(\cosh \varphi + \cos \theta),\\
p_2&=2-4\cosh\varphi\,\cos\theta,
\end{align*}
hence
\begin{align*}
\cosh\varphi&=\frac{1}{4}[p_1+\sqrt{p_1^2+4p_2-8}],\\
\cos\theta&=\frac{1}{4}[p_1-\sqrt{p_1^2+4p_2-8}].
\end{align*}
Under the above assumption, namely $\Lambda\ne I$, the result
$\theta=\varphi=0$ implies that case (b) applies.
\end{theorem}

\begin{proof}
It is sufficient to calculate the characteristic polynomials for
cases (a) and (b) of theorem \ref{nhp}, and to check the algebra.
\end{proof}

In case (b) all values $\alpha\ne 0$ correspond to the same
conjugacy class. For $\alpha \to 0$ the matrix representatives
converge to the identity, which belongs to a different conjugacy
class. As a consequence, the conjugacy class given by (b) is not
topologically closed in the topology of the Lie group. Thus, it is
impossible to distinguish between $\Lambda=I$, and case (b) with
$\Lambda\ne I$, by looking at the Ad-invariant continuous functions of
$\Lambda$.
%Indeed, in case (b) the class of conjugacy does not depend on
%$\alpha$ and yet the representative matrix tends to the identity for
%$\alpha\to 0$. Thus, any continuous invariant would assign the same
%value to $\alpha\ne 0$ and to $\alpha=0$.

The conjugacy classes of type (a) are topologically closed. Indeed,
the coefficients of the characteristic polynomial $p_1(\Lambda),
p_2(\Lambda)$, are polynomials in the matrix coefficients of
$\Lambda$, and hence are continuous in the Lie group topology.
Functions $\varphi(\Lambda),\theta(\Lambda)$, being continuous in
$p_1,p_2$, are also continuous with respect to the Lie group
topology. Each conjugacy class is determined by its value
$(\varphi,\theta)\in B:=[0,+\infty)\times [0,2\pi)$. The inverse
image of a $B$ point (which is closed) through the continuous map
$(\varphi \times \theta)(\Lambda)$ is a closed set, hence conjugacy
classes of type (a) are closed. Finally, the conjugacy class of the
identity is closed because it is just a point in the Lie group. In
summary.

\begin{proposition}
The only conjugacy class of the Lorentz group which is not
topologically closed is that of type (b). The closure of this class
contains the identity.
\end{proposition}

Through this same argument we can prove something more. Observe that
function $(\varphi \times \theta)(\Lambda)$  is invariant under
conjugation, thus so are the open  (closed) sets obtained as inverse
images of open (closed) sets. In particular, every distinct pair of
conjugacy classes of type (a) is separated by invariant open sets.

\section{The inhomogeneous Lorentz group}

When working on the affine space $M$, by {\em reference frame} we
shall mean an ordered pair $(o, \{e_a\})$ given by an {\em origin}
$o\in M$ and a proper orthochronous orthonormal base $\{e_a\}$. A
reference frame is then a point in the $SO(1,3)^{\uparrow}$-bundle
$R$ of reference frames \cite{kobayashi63}. Once a reference frame
has been chosen, any point $p\in M$ can be written in a unique way
in coordinates $\{x^a; a=0,1,2,3\}$, as $p=o+ x^a e_a$. For short,
from now on we shall denote the coordinate vector belonging to
$\mathbb{R}^4$ using a bar, e.g. $\bar{x}$.

As we mentioned in section \ref{tre},  a map $P\in
IL^{\uparrow}_{+}$ satisfies $P(p+w)=P(p)+\Lambda(w)$, for every
$p\in M$ and $w\in W$, where $\Lambda \in L^{\uparrow}_{+}$. The map
$P$ lifts to the bundle of reference frames as follows
\[
(o, \{e_a\}) \xrightarrow{P} (P(o), \{e_a'\}), \quad \textrm{where}
\ e_a'=\Lambda(e_a)=(\Lambda^{-1})^b_{\ a}
 e_b.
\]
Once a  reference frame has been chosen, an inhomogeneous proper
orthochronous Lorentz transformation $P$ induces a change of
coordinates
\[
{x^a}'=\Lambda^a_{\ b} x^b-b^a,
\]
where $\Lambda^a_b$ belongs to $SO(1,3)^\uparrow$ and
$P(o)-o=(\Lambda^{-1})^c_{\ d} b^d e_a$. We shall write this
coordinate transformation as
\begin{equation} \label{njw}
\begin{pmatrix} \bar{x}' \\ 1 \end{pmatrix} =\begin{pmatrix} \Lambda & -\bar{b} \\
\bar{0}^\intercal & 1 \end{pmatrix} \begin{pmatrix} \bar{x} \\ 1
\end{pmatrix}.
\end{equation} These matrices form
the group $ISO(1,3)^\uparrow$. The choice of reference frame
establishes an isomorphism between the  inhomogeneous Lorentz group
$IL^{\uparrow}_{+}$ and $ISO(1,3)^\uparrow$. A change of frame acts
as an automorphism $g\to c g c^{-1}$ of $ISO(1,3)^{\uparrow}$. We
remark that the equation $P(o)-o=(\Lambda^{-1})^c_{\ d} b^d e_a$
shows that $P(o)$ is in the causal (chronological) future of $o$ if
and only if $b^d$ is future  directed and causal (resp. timelike).

\subsection{The causal semigroup of $ISO(1,3)^{\uparrow}$}
\label{jnu}

The product of two elements of $ISO(1,3)^{ \uparrow}$ gives
\[
\begin{pmatrix} \Lambda_2 & -\bar{b}_2 \\
\bar{0}^\intercal & 1 \end{pmatrix} \begin{pmatrix} \Lambda_1 & -\bar{b}_1 \\
\bar{0}^\intercal & 1 \end{pmatrix}=\begin{pmatrix} \Lambda_2\Lambda_1 & -(\Lambda_2 \bar{b}_1+\bar{b}_2)  \\
\bar{0}^\intercal & 1 \end{pmatrix}
\]
as a consequence, if $\bar{b}_1$ and $\bar{b}_2$ are future directed
and nonspacelike  so is $\Lambda_2 \bar{b}_1+\bar{b}_2$. Thus we
have a semigroup on $ISO(1,3)^{\uparrow}$ which we call the (future
directed) {\em causal semigroup of} $ISO(1,3)^{\uparrow}$. We denote
it by $J$. Analogous considerations hold  for $\bar{b}$ f.d.\ and
timelike, and correspondingly we have  a timelike semigroup $I$.
Clearly, $I\subset J$. Notice that $J$, contrary to $I$, contains
the identity, thus it is a {\em monoid}. The set $J\cap J^{-1}$ is
the largest group contained in $J$ and it is isomorphic to
$SO(1,3)^{\uparrow}$ because it is made by those matrices for which
$\bar{b}=\bar{0}$.

%Given a semigroup $S$ of a group $G$, there is preorder $\le_S$
%associated to $S$ on $G$ given by: $x\le_S y$ iff $x^{-1} y\in S$
%(iff $y\in xS$). Using the expression for the inverse of an element
%in $ISO(1,3)^{\uparrow}$ we easily find that
%\[
%\begin{pmatrix} \Lambda_1 & -\bar{b}_1 \\
%\bar{0}^\intercal & 1
%\end{pmatrix} \le_J
%\begin{pmatrix} \Lambda_2 & -\bar{b}_2 \\
%\bar{0}^\intercal & 1 \end{pmatrix},   \qquad \textrm{iff} \
%\bar{b}_2-\bar{b}_1   \ \textrm{ is f.d.\ nonspacelike} \ (i.e. \
%\bar{b}_1\le \bar{b}_2).
%\]

The semigroups $I$ and $J$ are  important for the following reason.
Suppose that $P:M\to M$, $P\in IL^{\uparrow}_{+}$, is such that
there is a reference frame for which its matrix expression belongs
to $I$ (resp.  $J$). As we observed in section \ref{tre}, this means
that the origin $o$ of the frame that realizes the matrix reduction
is sent to $P(o)$ where $P(o)-o$ is f.d.\ timelike (resp. f.d.\
nonspacelike). At least in the timelike case we can interpret this
transformation as physically admissible. Indeed, we can select a
special point of space, namely the origin $o$, which moves forward
in time along a timelike geodesic segment. The transformation $P$
can then be interpreted as an active transformation induced by the
dragging of spacetime points along with this frame. Furthermore, we
known that the matrix expression changes by conjugacy under frame
changes. Thus, in order to find if a transformation $P$ falls into
this admissible class we have to find if the conjugacy class of the
matrix transformation, obtained in a generic frame, admits some
representative which belongs to $I$. We summarize this result with
the following proposition.

\begin{proposition} \label{ndz}
The map $P\in IL^{\uparrow}_+$ sends some point of $M$ in its
 causal (chronological) future if and only if  there is a
representative in the conjugacy class of its matrix representation
which belongs to the semigroup $J$ (resp. $I$).
\end{proposition}

\begin{proof}
Suppose that some point $o\in M$ is sent into its causal future.
Choose a frame at $o$, then the matrix representation of $P$ belongs
to $J$. The other direction has been proved above.
\end{proof}

%\begin{proof}
%Suppose that there is a representative in the conjugacy class which
%belongs to $J$. Since conjugacy transformation correspond to change
%of frames, there is a frame which induces the matrix representation
%given by such representative $x^a\to x^a'=\Lambda^a_{\ b} x^b-b^a$,
%where $b^a$ is f.d.\ causal. The origin of the target frame is
%$x^a'=0$ thus its position in the old coordinates if
%$(\Lambda^{-1})^a_{\ b} b^b$ which is still f.d.\ causal. Thus the
%origin $o$ of the frame is sent to its causal future.
%\end{proof}

While the group $ISO(1,3)^{\uparrow}$ is  the group of symmetries of the
spacetime manifold, the semigroup $I$ distinguishes itself as the
semigroup of symmetries that can be induced by the actual physical
movement of a frame on $M$. Here the elements   $ISO(1,3)^{\uparrow}$
 which have  to be discarded are those for which
there  is no point that is sent to its chronological future. The
elements in $E:=J\backslash I$ are rather special. These are maps
which respect causality at some point but which do not represent the
physical movement of a massive reference frame.

\subsection{The Lie algebra and its
interpretation}

The choice of reference frame establishes an isomorphism between
$IL^{\uparrow}_{+}$ and the group $ISO(1,3)^{\uparrow}$ made of
matrices $\begin{pmatrix} \Lambda & -\bar{b} \\ \bar{0}^{T} &
1\end{pmatrix}$ where $\Lambda \in SO(1,3)^{\uparrow}$ and
$\bar{b}\in \mathbb{R}^4$. Furthermore, it establishes a Lie algebra
isomorphism between the Lie algebra of $IL^{\uparrow}_{+}$,
$\mathfrak{IL}$, and the
Lie algebra  $\mathfrak{iso}(1,3)$ made of matrices $\begin{pmatrix} F& -\bar{w} \\
\bar{0}^{T} & 0\end{pmatrix}$, where $F \in
\mathfrak{so}(1,3)^{\uparrow}$ and $\bar{w}\in \mathbb{R}^4$, and
where the commutator in $\mathfrak{iso}(1,3)$ is the usual matrix
commutator.

Let us remind that $F\in \mathfrak{so}(1,3)^{\uparrow}$  iff
$F^{a}_{\ b}$ satisfies $F_{ab}+F_{ba}=0$, where the indices are
lowered using $\eta_{cd}$. This Lie algebra coincides with the Lie
algebra of the group $O(1,3)$ (because $SO(1,3)^{\uparrow}$ is the
connected component of $O(1,3)$ which contains the identity).

A significative base for $\mathfrak{iso}(1,3)$ is given by
\begin{align*}
J^{a b}&= \begin{pmatrix}M^{a b} & \bar{0}
\\\bar{0}^\intercal & 0
\end{pmatrix}, \\
P^a & = \begin{pmatrix} 0 & \bar{m}^a
\\ \bar{0}^\intercal & 0
\end{pmatrix} ,
\end{align*}
where
\begin{align*}
(M^{a b})^{c}_{\ d}&=\eta^{a c}
\delta^{b}_{d}-\eta^{b c} \delta^{a}_{d} , \\
(\bar{m}^a)^c&=\eta^{a c} .
\end{align*}
The subalgebra generated by $J^{a b}$ is the Lie algebra
$\mathfrak{so}(1,3)$. The non-vanishing commutation relations are
\begin{align}
{}[J^{a b}\!, J^{c d}] \!\!&=\! \!\eta^{a d} J^{b c}\!+\!\eta^{b c}
J^{c d}\!\!-\!\eta^{a c}J^{b d}\!\!-\!
\eta^{b d} J^{a c}, \\
[J^{a b}, P^{c}  ] \! &=\! \eta^{c b}P^{a} - \eta^{c a} P^{b} .
%[P^{\alpha}\!, P^{\beta}]&=0.
\end{align}
We shall write $H=P^0$. We introduce the generators
\begin{align*}
K^i&=J^{0 i}, \\
J^i&=\frac{1}{2}\epsilon_{ijk} \, J^{jk}, \qquad
(J^{jk}=\epsilon_{ijk}\, J^i).
\end{align*}
The non-vanishing commutation relations are (lowering space indices
does not introduce minus signs)
\begin{align*}
[J_i, J_j]&= -\epsilon_{ijk}\, J_k , \  &[J_i, K_j]&= - \epsilon_{ijk} \,K_k , \ &[K_i, K_j]&=  \epsilon_{ijk} \,J_k,\\
[J_i, P_j]&= - \epsilon_{ijk} \,P_k, \ &[K_i, P_j]&= \delta_{ij} \,
H , \ &[K_i, H]&=  P_i.
\end{align*}

The following matrix expressions clarify our conventions (which are
the same of \cite{minguzzi05e}).

{\footnotesize
\begin{align*}
  K^1&=\begin{pmatrix}
0 & -1 & 0 & 0 & 0 \\
-1 & 0 & 0 & 0 & 0 \\
0 & 0 & 0 & 0 & 0 \\
0 & 0 & 0 & 0 & 0 \\
0 & 0 & 0 & 0 & 0 \\
\end{pmatrix},  &J^3&= \begin{pmatrix}
0 & 0 & 0 & 0 & 0 \\
0 & 0 & 1 & 0 & 0 \\
0 & -1 & 0 & 0 & 0 \\
0 & 0 & 0 & 0 & 0 \\
0 & 0 & 0 & 0 & 0 \\
\end{pmatrix},  \\
 H&=
\begin{pmatrix}
0 & 0 & 0 & 0 & -1 \\
0 & 0 & 0 & 0 & 0 \\
0 & 0 & 0 & 0 & 0 \\
0 & 0 & 0 & 0 & 0 \\
0 & 0 & 0 & 0 & 0 \\
\end{pmatrix},  &P^1&=
\begin{pmatrix}
0 & 0 & 0 & 0 & 0 \\
0 & 0 & 0 & 0 & 1 \\
0 & 0 & 0 & 0 & 0 \\
0 & 0 & 0 & 0 & 0 \\
0 & 0 & 0 & 0 & 0 \\
\end{pmatrix}.
\end{align*}}

Arguing as in section \ref{kwi} we find that $\mathfrak{IL}$ is a
subalgebra of the algebra of vector fields on $M$, the
correspondence between matrices and vector fields being ($e_d=\p_d$)
\begin{equation} \label{ntv}
I+\epsilon [\frac{1}{2}\Omega_{a b}J^{a b}- b_{c} P^{c}] \
\xleftrightarrow{ \ \{e_a\}\ } \ [-\Omega^{d}_{\ b} x^b+ b^{d}]\,
e_{d},
\end{equation}

%The infinitesimal proper inhomogeneous Lorentz group transformation
%is generated by the matrix
%\begin{equation} \label{ntv}
%I+\frac{1}{2}\Omega_{\alpha \beta}J^{\alpha \beta}- b_{\gamma}
%P^{\gamma},
%\end{equation}
%where $\Omega_{\alpha \beta}, b_{\gamma}\sim \varepsilon$, for small
%$\varepsilon$.

By {\em observer} we mean a f.d.\ timelike worldline $\tau \to
\gamma(\tau)$, parametrized with respect to proper time and a
reference frame  $\{e_a\}(\tau)$ over it such that at any point
$e_0=\p_{\tau}$. If the tetrad is parallely transported the observer
is {\em inertial}. Starting from $\{e_a\}(\tau_0)$ one can regard
the motion of the observer as the repeated action of Lorentz
transformations $\delta \Lambda \in IL^\uparrow_{+}$ sending
$(\gamma(\tau), \{e_a\}(\tau))$ to $(\gamma(\tau+\delta
\tau),\{e_a\}(\tau+\delta \tau))$. At each instant we have
coordinates $x^a_{\tau}$ associated to the frame $\{e_a\}(\tau)$,
thus $\delta\Lambda(\tau)$ induces a change of coordinates. Let
$\vec{a}$ be the acceleration of the observer, and let
$\vec{\omega}$ be its angular velocity (all quantities are measured
by herself). The coordinate change induced by the motion of the
observer  in a proper time interval $\dd \tau$ is
\begin{equation} \label{njb}
I+(a_i K^i+\omega_k J^k+ H) \dd \tau.
\end{equation}
as it can be easily inferred from its matrix form (see
\cite{minguzzi05e} for another argument). In this equation
$a_i=\Omega_{0i}$, and $\omega_k =
\frac{1}{2}\epsilon_{kij}\Omega_{ij}$, where we have identified the
small parameter $\epsilon$ with $\dd \tau$. Thus the infinitesimal
motion of the observer is given by  matrix (\ref{ntv}) for $b^0 =1$
and $b^i=0$ for $i=1,2,3$. The vector field which generates the
infinitesimal transformation of $M$, and hence the change of
reference frame to which correspond the coordinate change
(\ref{njb}), is
\begin{equation}
\p_0+a_i(x^0\p_i+x^i\p_0)+\omega_k(\epsilon_{kij}x^i\p_j)
\end{equation}
where we used the corresponences
\begin{align*}
K_i & \leftrightarrow x^0\p_i+x^i\p_0, &J_i & \leftrightarrow
\epsilon_{kij}x^i\p_j,  \\
H & \leftrightarrow \p_0 , &P_i & \leftrightarrow -\p_i.\\
\end{align*}
It is easy to check that these vector fields  satisfy the same
commutation relations of their matrix counterparts.

%where $\p_0$ corresponds to $H$, $(x^0\p_i+x^i\p_0)$ to $K_i$,
%$(\epsilon_{kij}x^i\p_j)$ to $J_i$, and where $-\p_i$, which
%corresponds to $P_i$, does not appear.

 It is natural to ask why in the full inhomogeneous Lorentz
group we have to consider translations generated by $P^i$ if they do
not appear as Lie algebra generators of the observer's movements.
The answer is that the operators $P^i$ arise at the
non-infinitesimal level, through the compositions of several
operations of the above type, as a consequence of the Lie algebra
commutation relations (e.g.\ through Baker-Campbell-Hausdorff
formula).

%Suppose that $b^0 >0$ and $b^i=0$ for $i=1,2,3$. It %is convenient to
%define $\dd \tau= b^0$, $a_i=\Omega_{0 i}/ \dd \tau $,
%$\omega_k=\frac{1}{2} \epsilon^{ijk} \Omega_{ij} /\dd \tau$, and
%rewrite the previous expression as
%\begin{equation}
%I+(a_i K^i+\omega_k J^k+ H) \dd \tau.
%\end{equation}
%Imagine an observer that accelerates with proper acceleration
%$\vec{a}$ and angular velocity $\vec{\omega}$ for  a proper time
%interval $\dd \tau$ (all quantities are measured by the non-inertial
%observer). Consider the curve of frames $\{e_a\}(\tau)$ attached to
%its worldline (Fermi-Walker transport \cite{misner73}). The matrix
%expression of the previous infinitesimal transformation shows that
%it is the coordinates transformation between the initial and final
%inertial coordinates systems associated to the initial
%$\{e_a\}(\tau)$ and final  $\{e_a\}(\tau+\dd \tau)$ frames of the
%observer (see \cite{minguzzi05e} for another argument).

%Since this is the most general motion that one can perform,

\subsection{Exponentiality of $ISO(1,3)^\uparrow$}

Let us consider the equation which defines the exponential map
\[
\frac{\dd }{\dd s}
\begin{pmatrix}
\Lambda & - \bar{b} \\
\bar{0}^\intercal & 1
\end{pmatrix}
=
\begin{pmatrix}
F & - \bar{w} \\
\bar{0}^\intercal & 0
\end{pmatrix}
\begin{pmatrix}
\Lambda & - \bar{b} \\
\bar{0}^\intercal & 1
\end{pmatrix},
\]
with initial condition $\Lambda=I$, $\bar{b}=\bar{0}$.  The matrix
equation  is equivalent to the system
\begin{align}
\frac{\dd }{\dd s} \,\Lambda &=F\Lambda,  \label{ndg}\\
\frac{\dd }{\dd s} \,\bar{b} &=F\bar{b}+\bar{w}. \label{ndf}
\end{align}
Let $\bar{c}(s)$ be such that $\bar{b}(s)=\Lambda(s)\bar{c}(s)$
(thus $\bar{b}$ is f.d.\ nonspacelike iff $\bar{c}$ is). Equation
(\ref{ndf}) is equivalent to
\begin{equation} \label{lny}
\frac{\dd }{\dd s} \,\bar{c}(s)=\Lambda^{-1}(s) \bar{w}.
\end{equation}
Through these equations, and using the exponentiality of
$SO(1,3)^{\uparrow}$,  we are now able to prove.

\begin{theorem} \label{biq}
The group $ISO(1,3)^\uparrow$ is exponential. More in detail, for
every $F\in \mathfrak{so}(1,3)$ the matrix $(e^F-I)/F$ is well
defined and invertible. Let ${\footnotesize \begin{pmatrix}
\Lambda & - \bar{b} \\
\bar{0}^\intercal & 1
\end{pmatrix}} \in ISO(1,3)^\uparrow$ (that is, $\Lambda \in
SO(1,3)^{\uparrow}$), then there is some  $F\in \mathfrak{so}(1,3)$
such that $\Lambda=\exp F$ and for any such choice the vector
$\bar{w}$ defined by
\[
\bar{w}=\frac{F}{e^F-I} \,  \bar{b}
\]
is such that, $\exp[ {\footnotesize \begin{pmatrix}
F & - \bar{w} \\
\bar{0}^\intercal & 0
\end{pmatrix}}]= {\footnotesize \begin{pmatrix}
\Lambda & - \bar{b} \\
\bar{0}^\intercal & 1
\end{pmatrix}}$ and is the only vector with this property.
Furthermore, $\bar{w}$ is a f.d.\ null eigenvector of $F$ with
eigenvalue $\lambda\in\mathbb{R}$, if and only if $\bar{b}$ is a
f.d.\ null eigenvector of $\Lambda$ with eigenvalue $\exp \lambda
\in (0,+\infty)$, in which case $\bar{w}=\frac{\lambda}{e^\lambda-1}
\bar{b}$.
%Apart from this special case, if $\bar{w}$ is f.d.\ causal
%then $\bar{b}$ is f.d.\ timelike.
\end{theorem}

\begin{proof}

The function $f(z)=(e^z-1)/z$ is analytic thus it makes sense to
consider  $f(F)$, for $F\in \mathfrak{so}(1,3)$. Since $f$ satisfies
$f(LFL^{-1})=Lf(F) L^{-1}$ for $L\in SO(1,3)^\uparrow$, in order to
prove its invertibility we have just to verify this property over
one representative for each orbit on $\mathfrak{so}(1,3)$.
Therefore, we can  use the  representatives  $A$ and $B$ selected in
theorem \ref{vkx}.

If $F=A$ then
\[
(e^A-I)/A=\begin{pmatrix}
 \frac{\sinh \varphi}{\varphi} &  \frac{1-\cosh\varphi}{\varphi} & 0 & 0 \\
 \frac{1-\cosh\varphi}{\varphi}  & \frac{\sinh \varphi}{\varphi} & 0  & 0 \\
 0  & 0 & \frac{\sin \theta}{\theta} & \frac{1-\cos \theta}{\theta}\\
 0 & 0 & \frac{1-\cos \theta}{\theta} & \frac{\sin \theta}{\theta}
\end{pmatrix},
\]
which has positive determinant $\frac{2(\cosh \varphi
-1)}{\varphi^2} \frac{2(1 -\cos\theta)}{\theta^2}$ (this expression
makes sense and is finite for $\varphi=0$ or $\theta=0$).

If $F=B$ we can use $B^3=0$ so that
\[
(e^B-I)/B=I+\frac{B}{2}+\frac{B^2}{6} =
\begin{pmatrix}
 1+ \alpha ^2/6 &  -\alpha^2/6 & -\alpha/2 & 0 \\
 \alpha^2/6 & 1-\alpha ^2/6 & -\alpha/2  & 0 \\
 -\alpha/2 & \alpha/2 & 1 & 0 \\
 0 & 0 & 0 & 1
\end{pmatrix},
\]
which has determinant equal to 1.

Let us try to find $\bar{w}$ in such a way that  $\exp[
{\footnotesize \begin{pmatrix}
F & - \bar{w} \\
\bar{0}^\intercal & 0
\end{pmatrix}}]= {\footnotesize \begin{pmatrix}
\Lambda & - \bar{b} \\
\bar{0}^\intercal & 1
\end{pmatrix}}$.
 From Eq. (\ref{lny}) it follows
\[
\bar{c}(1)=[\int_0^1 \Lambda^{-1}(s) \dd s] \bar{w}
\]
and we must comply with $\bar{b}(1)=\Lambda(1) \bar{c}(1)$. But the
solution to Eq. (\ref{ndg}) is $\Lambda(s)=\exp[F s]$, and by
assumption $\Lambda=\Lambda(1)=\exp F$, thus
\[
\bar{b}= (\exp F) \bar{c}(1)= \exp F [\int_0^1 e^{-F s} \dd s]
\bar{w}=\frac{\exp F-I}{F} \,\bar{w}
\]
Since the matrix on the right-hand side is invertible there is one
and only one vector $\bar{w}$ which complies with this equation.

Finally, the last statement  follows easily from the fact that
$\Lambda$ and $F$ have the same eigenvectors (Prop. \ref{nux}).

%In case $F=B$ we have the identity
%\[
%(e^B-I)/B=e^{(F/2)}+\frac{F^2}{24}
%\]
%DA COMPLETARE]
\end{proof}

\subsection{Ad-invariants and Lie algebra orbits} \label{bos}

Let us consider the Ad-action of $ISO(1,3)^{\uparrow}$ on
$\mathfrak{iso}(1,3)$,
\[
\begin{pmatrix} L & -\bar{a} \\
  \bar{0}^\intercal & 1 \end{pmatrix}
\begin{pmatrix}
F & - \bar{w} \\
\bar{0}^\intercal & 0
\end{pmatrix}
\begin{pmatrix} L & -\bar{a} \\
  \bar{0}^\intercal & 1 \end{pmatrix}^{-1}=\begin{pmatrix}
L F L^{-1} & LFL^{-1} \bar{a}- L\bar{w} \\
\bar{0}^\intercal & 0
\end{pmatrix}.
\]
We shall be interested in the separated effect of
\begin{itemize}
\item[(i):] homogeneous transformations of the frame
\begin{equation}
\begin{array}{l}
F\to L F L^{-1},\\
\bar{w}\to L\bar{w},
\end{array} \label{hom}
\end{equation}
\item[(ii):] translations of the frame
\begin{equation}
\begin{array}{l}
F\to F,\\
\bar{w}\to \bar{w}-F\bar{a}.
\end{array} \label{tra}
\end{equation}
\end{itemize}
The  action on the homogeneous part $F$ coincides with the Ad-action
of $SO(1,3)^{\uparrow}$ on $\mathfrak{so}(1,3)$.
%
% be an element of $\mathfrak{iso}(1,3)$, then the
%\begin{align}
%F &\to L F L^{-1}\\
%w\to Lw
%\end{align}
%
%mentioned action acts on the homogeneous part $F$ as the Ad-action
%of $SO(1,3)^{\uparrow}$ on $\mathfrak{so}(1,3)$.
Thus the invariants $I_1$ and $I_2$ of section \ref{igk} are still
invariants for the Ad action on the inhomogeneous Lie algebra. We
are going to show that if $I_2=0$ then there is a third invariant
$I_3$. It represents a kind of relativistic generalization of the
square of the screw scalar $\langle v,v\rangle$ (unfortunately,
there does not seem to be any convenient relativistic generalization
of the screw product).

It will be convenient to keep in mind that the most generic frame
transformation can be accomplished through a translation followed by
a homogeneous transformation, according to this scheme
\begin{align*}
F&\to  F  \to LFL^{-1},\\
\bar{w}&\to \bar{w}-F (L^{-1}\bar{a})\to L[\bar{w}-F
(L^{-1}\bar{a})].
\end{align*}

\begin{lemma} \label{jef}
Let $F\in \mathfrak{so}(1,3)$,  and let $\tilde{F}_{c d}=\frac{1}{2}
\varepsilon_{a b c d} F^{a b}$, then
\[
\tilde{F}_{c d} F^{c e}=\frac{1}{4} (F^{a b} \tilde{F}_{a b})
\delta^e_{d}
\]

\end{lemma}

\begin{proof}
\begin{align*}
\tilde{F}_{\gamma \delta} F^{\gamma \sigma}&=\frac{1}{2}
\varepsilon_{\alpha \beta \gamma \delta} F^{\alpha \beta} F^{\gamma
\sigma}=\frac{1}{2} \varepsilon_{\alpha \beta \gamma \delta}
F^{\alpha \beta} [-\frac{1}{2} \varepsilon^{\gamma \sigma \eta \nu}
\tilde{F}_{\eta \nu} ]=-\frac{1}{4} F^{\alpha \beta} \tilde{F}_{\eta
\nu} (-\delta^{\sigma\eta \nu}_{\alpha \beta \delta})\\
&= \frac{1}{4} F^{\alpha \beta} \tilde{F}_{\eta \nu} (\delta^{\eta
\nu}_{\alpha \beta}\delta^\sigma_\delta+\delta^{\sigma \eta}_{\alpha
\beta} \delta^\nu_\delta+\delta^{\nu \sigma}_{\alpha
\beta}\delta^\eta_\delta)=\frac{1}{2} (F^{\eta \nu} \tilde{F}_{\eta
\nu} \delta^\sigma_\delta +F^{\sigma \eta} \tilde{F}_{\eta \delta}+
F^{\nu \sigma} \tilde{F}_{\delta \nu})\\
&=-\tilde{F}_{\gamma \delta} F^{\gamma \sigma}+\frac{1}{2}
(F^{\alpha \beta} \tilde{F}_{\alpha \beta}) \delta^\sigma_\delta
\end{align*}
\end{proof}

\begin{theorem} \label{gjs}
Let $\begin{pmatrix}
F & - \bar{w} \\
\bar{0}^\intercal & 0
\end{pmatrix}$ be an element of $\mathfrak{iso}(1,3)$ and let us
suppose that $I_2:=-\frac{1}{2} \tilde{F}_{c d} F^{c d}=0$, then
\[
I_3:=  \tilde{F}_{a b} w^b \tilde{F}^{a}_{\ \, c} w^c
\]
is Ad-invariant.

Condition $F\ne 0$ is Ad-invariant. If $I_2=0$ the condition
$\tilde{F} \bar{w}\ne 0$ is Ad-invariant. If $I_2=0$, $I_1>0$,
$I_3\ge 0$, $\tilde{F} \bar{w}\ne 0$, then $\epsilon_1 =\textrm{sgn}
\,( \tilde{F}^0_{\ a} \tilde{F}^{a}_{\ b} w^b)$ is
 a well defined Ad-invariant. If $I_2=0$, $I_1>0$, $I_3\le 0$, $\tilde{F} \bar{w}\ne
 0$,
then $\epsilon_2=\textrm{sgn}\, (\tilde{F}^{0}_{\ \, c} w^c)$ is a
well defined  Ad-invariant.

If $I_2=I_1=0$ but $F\ne 0$, then $I_3$ is non-negative. The
equality $I_3=0$ holds if and only if $\bar{w}\in Ker F^2$ (this
statement is Ad-invariant). Let us consider the cases $I_3>0$ and
$I_3=0$.

If $I_3>0$ then $\epsilon_1= \textrm{sgn}( \tilde{F}^0_{\ a}
\tilde{F}^{a}_{\ b} w^b)$ is well defined and provides another
Ad-invariant.

Let us come to $I_3=0$. Since $F^3=0$, we have $\textrm{Im} F\subset
\textrm{Ker} F^2$. Frame translations send $\bar{w}$ to another
element in the same class of $\textrm{Ker} F^2/\textrm{Im} F$. There
is a choice that minimizes ${w'}^a w'_a$. This (non-negative)
minimum
\[
I_4:=\min_{\bar{w}'\in [\bar{w}]} {w'}^a {w'}_a,
\]
is a characteristic of the class and is, therefore, an Ad-invariant.
(Any minimizing element belongs to $\textrm{Ker} F$.)

Finally, if $F=0$ (and hence $I_2=I_1=0$), then \[I_4:=w^a w_a,\] is
Ad-invariant. Furthermore,  if $I_4<0$, then
$\epsilon_1=\textrm{sgn}\, w^0$ is Ad-invariant.
\end{theorem}

\begin{remark}
The equation defining  $I_4$ for $F=0$ follows from an extension of
that defining $I_4$ for $I_1=I_2=I_3=0$, $F\ne 0$. Indeed, if $F=0$,
$\textrm{Ker} F^2=W$ and $\textrm{Im} F=\{0\}$ thus every class
contains only one element, i.e. $\textrm{Ker} F^2/\textrm{Im} F =W$,
thus $\min_{\bar{w}'\in [\bar{w}]} {w'}^a {w'}_a=w^a w_a$.
\end{remark}

\begin{proof}
It is sufficient to prove that $I_2$ is invariant under (i)
homogeneous transformations, i.e., $\bar{a}=\bar{0}$, and (ii)
translations, i.e., $L=I$. Case (i) is clear since $I_3$ is defined
as the (Lorentzian) square of a 4-vector and, under the assumption
$\bar{a}=\bar{0}$, both $F^{a}_{\ b}$ and $w^{b}$ transform as
tensors. For case (ii) observe that $F\to F$ and $w^b\to w^b-F^b_{\
c} a^c$, and by lemma \ref{jef}, $I_3$ is left invariant.

Condition $F\ne 0$ is trivially Ad-invariant. If $I_2=0$ condition
$\tilde{F}\bar{w}\ne 0$ is invariant under frame changes because of
lemma \ref{jef}.  Suppose that $I_2=0$, $I_1>0$, $I_3\ge 0$, and
$\tilde{F}\bar{w}\ne 0$ if $I_3=0$.
 Let $z^c:= \tilde{F}^c_{\ a} \tilde{F}^{a}_{\ b}
w^b$. There is a frame for which $z^0  \ne 0$, and $z^\gamma$ is
null, namely that for which $F=A$ with $\varphi=0$. Under
translations of the frame $z^c$ does not change, while under
homogeneous transformations it behaves as a vector, thus it is a
lightlike vector in any frame and $\epsilon_1:= \textrm{sgn} (z^0)$
is invariant.

%Let us consider how $\epsilon_1:= \textrm{sgn} (z^0)$ changes under
%frame changes. In case (i), since $w^\alpha$ transforms as a
%4-vector the sign of $F^0_{\ \alpha} F^{\alpha}_{\ \beta} w^\beta$
%does not change. In case (ii), $F^0_{\ \alpha} F^{\alpha}_{\ \beta}
%w^\beta$ does not change because of lemma \ref{jef}.

%
%is left invariant because
%
%
%
%
%let $F=CAC^{-1}$ where $A$ is the canonical form of theorem
%\ref{vkx} with $\varphi=0$, and $C\in SO(1,3)^\uparrow$. Then the
%map $\bar{w}\to \bar{w}-F \bar{a}$, reads $\bar{w} \to C^{-1}\bar{w}
%\to C^{-1} \bar{w} -A (C^{-1} \bar{a}) \to C[C^{-1} \bar{w} -A
%(C^{-1} \bar{a}) ]$ and none of the steps changes $\textrm{sgn} \,
%w^0$.

Suppose  that $I_2=0$, $I_1>0$, $I_3\le 0$, and $\tilde{F}\bar{w}\ne
0$ if $I_3=0$. Let $v^a:=\tilde{F}^{a}_{\ \, c} w^c$. Since $I_3\le
0$, there is a frame for which $v^0\ne 0$,
 namely that for which $F=A$ with $\varphi=0$. In that frame, since
$I_3\le 0$, $v^c$ is a causal vector. Because of lemma \ref{jef},
under translations of the frame $v^c$ is left unaltered, while under
homogeneous transformations it transform as a tensor, thus in any
frame $v^c$ is a causal vector and $\epsilon_2:= \textrm{sgn}\, v^0$
is invariant.

%The sign $\epsilon_2$ is clearly Ad-invariant because under (i) the
%vector $v^\alpha:=\tilde{F}^{\alpha}_{\ \, \gamma} w^\gamma$
%transforms as a tensor, while under (ii) $v^\alpha$ is left
%unaltered because of lemma \ref{jef} and $I_2=0$.

Suppose $I_2=I_1=0$ and $F\ne 0$. The inequality $I_3 \ge 0$ can be
easily checked for $F=B$ and arbitrary $\bar{w}$. Since $I_3$ is a
scalar under transformations of type (i), the inequality holds for
any $F$ in the orbit of $B$. With the same type of argument, i.e.
studying the case $F=B$, we can show that $\tilde{F}^2=F^2$, and
from the special form of $B^2$ we easily deduce that $I_3=0$ if and
only if $\bar{w}\in \textrm{Ker} F^2$.

% and that if $I_3\ne 0$ then
%in the frame for which $F=B$, $F^0_{\ \alpha} F^{\alpha}_{\ \beta}
%w^\beta\ne 0$.

Let us observe that the condition $\bar{w}\in \textrm{Ker} F^2$ is
independent of the frame since under translations of the frame
$\bar{w}$ is added terms belonging to $\textrm{Im} F$, and
$\textrm{Im} F\subset \textrm{Ker} F^2$ as $F^3=0$.

Let us consider $\epsilon_1:= \textrm{sign}\, v^0$ where
$v^c:=\tilde{F}^c_{\ a} \tilde{F}^{a}_{\ b} w^b=F^c_{\ a} F^{a}_{\
b} w^b$. Let us observe that $v^c$ is a null vector because $F^3=0$.
The frame for which $F=B$ shows that if $\bar{w}\notin \textrm{Ker}
F^2$ (iff $I_3>0$) then $v^0\ne 0$ in that frame. Under homogeneous
transformations of the frame $v^c$ transform as a tensor (hence
remaining a null vector), while under translations of the frame it
is left unchanged because $F^3=0$. As a consequence, $\epsilon_1$ is
well defined and invariant.

%The variable $\epsilon_3$ is Ad-invariant because under (i) the
%vector $v^\gamma$ transform as a tensor, while under (ii) $v^\alpha$
%is left unaltered because $F^3=0$ whenever $F$ belongs to class (b)
%as given in theorem \ref{vkx}. Furthermore

Let us consider the possibility $I_3=0$, and hence $\bar{w}\in
\textrm{Ker} F^2$. We mentioned that under frame changes $\bar{w}$
transforms as $\bar{w}\to \bar{w}-F\bar{a}$, thus being altered by
additive terms belonging to $\textrm{Im} F$. Let us show that this
additive term can be chosen so as to minimize $w^a w_a$. Let us
study the problem in the frame for which $F=B$ so that:
$\textrm{Ker} F=\textrm{Span}(e_0+e_1,e_3)$, $\textrm{Ker}
F^2=\textrm{Span}(e_0+e_1,e_2,e_3)$, $\textrm{Im}
F=\textrm{Span}(e_0+e_1,e_2)$, $\textrm{Im}
F^2=\textrm{Span}(e_0+e_1)$. Then $w=(l,l,c,d)$, $w^a w_a=c^2+d^2$,
where $l$ and $c$ can be chosen freely. Clearly the minimum exists,
and is attained for $c=0$ ($l$ remains undetermined). Any minimizing
vector, being of the form $w=(l,l,0,d)$, belongs to $\textrm{Ker}
F$.

The last statement is trivial.

\end{proof}

\begin{definition}
We call  $\epsilon_1$ the {\em time orientation} of the Lie algebra
element (future directed or positive if $\epsilon_1=+1$). We call
$\epsilon_2$  {\em helicity}.
\end{definition}

%
%\begin{lemma}
%If $I_2=I_1=0$ but $F\ne 0$, then $I_3=0$ if and only if $\bar{w}\in
%Ker F^2$ (this statement is invariant under frame changes). Since
%$F^3=0$, we have $\textrm{Im} F\subset \textrm{Ker} F^2$. Frame
%translations send $\bar{w}$ to another element still belonging to
%the same class of $\textrm{Ker} F^2/\textrm{Im} F$. There is choice
%that minimizes $I_4:=w_\alpha w^\alpha$. This minimum is a
%characteristic of the class and is, therefore, an Ad-invariant.
%\end{lemma}
%
%
%\begin{proof}
%
%\end{proof}

\begin{theorem} (Classification of Lie orbits)
\label{vkz}\\
Let $\mathcal{P}\in \mathfrak{IL}$, let $I_1,I_2,I_3,I_4,
\epsilon_1,\epsilon_2,$ be the Ad-invariants of $\mathcal{P}$ and
let  $\varphi$ and $\theta$ be defined as in Eqs.
(\ref{vt1})-(\ref{vt2}). Moreover, if $I_2=0$ and $I_1\ne 0$ let
\[
b=\sqrt{\vert \frac{I_3}{2 I_1}\vert }.
\]
Then it is possible to choose the reference frame in such a way that
$\mathcal{P}$  takes one of the following matrix forms, and
corresponding vector field form.
\begin{enumerate}
\item $I_2\ne 0$:
\[
 {\footnotesize
\begin{pmatrix}
0 & -\varphi  & 0 & 0 & 0 \\
 -\varphi & 0 & 0 & 0 & 0 \\
 0 & 0 & 0 & \theta  & 0 \\
 0 & 0 & -\theta & 0 & 0 \\
 0 & 0 & 0 & 0 & 0
\end{pmatrix}}, \qquad \varphi(x^0\p_1+x^1\p_0)+\theta
(x^2\p_3-x^3\p_2),
\]
\item $I_2=0$, $I_1< 0$ ($I_3\ge 0$):
\[
{\footnotesize
\begin{pmatrix}
0 & -\varphi  & 0 & 0 & 0 \\
 -\varphi & 0 & 0 & 0 & 0 \\
 0 & 0 & 0 & 0  & 0 \\
 0 & 0 & 0 & 0 & -b \\
 0 & 0 & 0 & 0 & 0
\end{pmatrix}}, \qquad \qquad \qquad \varphi(x^0\p_1+x^1\p_0)+b \p_3,
\]
\item $I_2=0$, $I_1> 0$, $I_3< 0$:
\[
{\footnotesize
\begin{pmatrix}
0 & 0  & 0 & 0 & 0 \\
0 & 0 & 0 & 0 & -\epsilon_2 b \\
 0 & 0 & 0 & \theta & 0 \\
 0 & 0 & -\theta & 0 & 0 \\
 0 & 0 & 0 & 0 & 0
\end{pmatrix}}, \qquad \qquad \quad  \ \theta
(x^2\p_3-x^3\p_2)+ \epsilon_2 b \p_1,
\]
\item $I_2=0$, $I_1> 0$, $I_3> 0$:
\[
{\footnotesize
\begin{pmatrix}
0 & 0  & 0 & 0 & -\epsilon_1 b \\
0 & 0 & 0 & 0 & 0 \\
 0 & 0 & 0 & \theta & 0 \\
 0 & 0 & -\theta & 0 & 0 \\
 0 & 0 & 0 & 0 & 0
\end{pmatrix}}, \qquad \qquad \quad  \theta
(x^2\p_3-x^3\p_2)+ \epsilon_1 b \p_0,
\]
\item $I_2=0$, $I_1> 0$, $I_3=0$, $\tilde{F} \bar{w}\ne
 0$:
\[
{\footnotesize
\begin{pmatrix}
0 & 0  & 0 & 0 & -\epsilon_1  \\
0 & 0 & 0 & 0 & -\epsilon_2 \\
 0 & 0 & 0 & \theta & 0 \\
 0 & 0 & -\theta & 0 & 0 \\
 0 & 0 & 0 & 0 & 0
\end{pmatrix}}, \qquad \theta
(x^2\p_3-x^3\p_2)+ \epsilon_1  \p_0 + \epsilon_2  \p_1,
\]
\item $I_2=0$, $I_1> 0$, $I_3=0$, $\tilde{F} \bar{w}=0$:
\[
{\footnotesize
\begin{pmatrix}
0 & 0  & 0 & 0 & 0  \\
0 & 0 & 0 & 0 & 0 \\
 0 & 0 & 0 & \theta & 0 \\
 0 & 0 & -\theta & 0 & 0 \\
 0 & 0 & 0 & 0 & 0
\end{pmatrix}}, \qquad \qquad \qquad \qquad \theta (x^2\p_3-x^3\p_2),
\]
\item $I_2=I_1=0$, $F\ne 0$, $I_3>0$:
\[ \
{\footnotesize
\begin{pmatrix}
0 & 0  & -1 & 0 & - \epsilon_1 \sqrt{I_3}  \\
0 & 0 & -1 & 0 & 0 \\
 -1 & 1 & 0 & 0 & 0 \\
 0 & 0 & 0 & 0 & 0 \\
 0 & 0 & 0 & 0 & 0
\end{pmatrix}}, \quad (x^0\p_2+x^2\p_0)+
(x^2\p_1-x^1\p_2)+\epsilon_3\sqrt{I_3}\, \p_0,
\]
\item $I_2=I_1=0$, $F\ne 0$, $I_3=0$:
\[ \
{\footnotesize
\begin{pmatrix}
0 & 0  & -1 & 0 & 0  \\
0 & 0 & -1 & 0 & 0 \\
 -1 & 1 & 0 & 0 & 0 \\
 0 & 0 & 0 & 0 & -\sqrt{I_4} \\
 0 & 0 & 0 & 0 & 0
\end{pmatrix}}, \quad (x^0\p_2+x^2\p_0)+
(x^2\p_1-x^1\p_2)+\sqrt{I_4}\, \p_3,
\]
%where $\alpha\ne 0$ and all the values of $\alpha$ with the same
%sign belong to the same orbit (thus we can choose $\alpha=\pm 1$).
\item  $F=0$, $I_4<0$:
\[
{\footnotesize
\begin{pmatrix}
0 & 0  & 0 & 0 & - \epsilon_1 \sqrt{-I_4}  \\
0 & 0 & 0 & 0 & 0 \\
 0 & 0 & 0 & 0 & 0 \\
 0 & 0 & 0 & 0 & 0 \\
 0 & 0 & 0 & 0 & 0
\end{pmatrix}}, \qquad \qquad \qquad  \epsilon_1
\sqrt{-I_4}\,\p_0, \quad \ \ \quad \quad \qquad \qquad \empty
\]
\item  $F=0$, $I_4>0$:
\[
{\footnotesize
\begin{pmatrix}
0 & 0  & 0 & 0 & 0  \\
0 & 0 & 0 & 0 &  0 \\
 0 & 0 & 0 & 0 & 0 \\
 0 & 0 & 0 & 0 & -\sqrt{I_4} \\
 0 & 0 & 0 & 0 & 0
\end{pmatrix}}, \qquad \qquad \qquad \sqrt{I_4}\,\p_3, \quad \  \qquad \qquad \qquad \empty
\]
\item  $F=0$, $I_4=0$, $\mathcal{P}\ne 0$:
\[
\qquad {\footnotesize
\begin{pmatrix}
0 & 0  & 0 & 0 & - \epsilon_1  \\
0 & 0 & 0 & 0 & - \epsilon_1 \\
 0 & 0 & 0 & 0 & 0 \\
 0 & 0 & 0 & 0 & 0 \\
 0 & 0 & 0 & 0 & 0
\end{pmatrix}}, \ \qquad \qquad \qquad  \epsilon_1(\p_0+\p_1). \quad \ \ \qquad \quad \qquad \empty
\]
\end{enumerate}
Finally, there is a twelfth case corresponding to the trivial  Lie
algebra orbit of the zero element of $\mathfrak{IL}$.

Stated in another way, the orbits of the adjoint action of
$ISO(1,3)^\uparrow$ on $\mathfrak{iso}(1,3)$ admit one and only one
of the twelve representatives given above.

\end{theorem}

\begin{proof}
If $F=0$ then $w^\alpha$ transform as a vector under changes of
frame. Cases 9, 10 and 11 are then rather obvious, it is sufficient
to observe that under a rotation of the frame we can accomplish
$w^2=w^3=0$, thus though a boost on the timelike plane
$Span(e_0,e_1)$ we can obtain one of the forms 9, 10 or 11 (or the
trivial zero element).

Thus let $F\ne 0$. If $I_2\ne 0$ then $F$ is non-singular thus we
find a suitable translation $\bar{w}\to \bar{w}-F\bar{a}$ which
sends $\bar{w}$ to zero. Then with a homogeneous transformation we
send $F$ to the canonical form $A$ of theorem \ref{vkx}. We obtain
in this way the representative of case 1.

Suppose that $F\ne 0$, $I_2=0$, $I_1\ne 0$. According to theorem
\ref{vkx} through a homogeneous transformation of the reference
frame we can send $F$ to the canonical form $A$. If $I_1<0$ then
$\theta=0$, if $I_1>0$ then $\varphi=0$.

In the former case $A\vert_{\textrm{Span}(e_0,e_1)}:
\textrm{Span}(e_0,e_1) \to \textrm{Span}(e_0,e_1)$ is non-singular
thus with a translation of the reference frame $\bar{w}\to
\bar{w}-A\bar{a}$ we accomplish $w^0=w^1=0$. With a rotation of the
reference plane on the plane $\textrm{Span}(e_2,e_3)$ we obtain
$w^2=0$. Finally, with a rotation of $\pi$ along the first axis we
choose suitably the sign of $w^3$ so as to send it to $-b$ (by the
existence of the invariants). We arrive in this way at the
representative 2.

Let us consider the latter case $I_1>0$, $\varphi=0$. The map
$A\vert_{\textrm{Span}(e_2,e_3)}: \textrm{Span}(e_2,e_3) \to
\textrm{Span}(e_2,e_3)$ is non-singular, thus with a translation of
the reference frame $\bar{w}\to \bar{w}-A\bar{a}$ we accomplish
$w^2=w^3=0$. Then with a boost in the timelike plane
$\textrm{Span}(e_0,e_1)$ we accomplish one of the representatives
3,4, 5 or 6.

Suppose that $I_2=I_1=0$, $\bar{w}\notin \textrm{Ker} F^2$. With a
homogeneous transformation of the reference frame we  send $F$ to
$B$ for some $\alpha>0$. The image of $B$ is
$\textrm{Span}(e_0+e_1,e_2)$ thus with a translation of the
reference frame we obtain $w^1=w^2=0$. Since $F=B$ the condition
$\bar{w}\notin \textrm{Ker} F^2$ implies $w^0\ne 0$.

With a boost in the timelike plane $\textrm{Span}(e_0,e_1)$ of
rapidity $r$ followed by a translation we send $\alpha$ to
$\alpha':=\alpha e^{-r}$ and $w^0$ to $w^0 e^r$, keeping
$w^1=w^2=0$. Thus we can choose $r$ so that $\vert w^0e^r \vert >
\vert w^3\vert$. Next we use the identity
\begin{align*}
  & {\footnotesize\begin{pmatrix}
 \cosh \gamma & 0 & 0 & -\sinh \gamma & 0 \\
 0 & 1 & 0 & 0 & 0 \\
 0 & 0 & 1 & 0 & 0\\
 -\sinh \gamma & 0 & 0 & \cosh \gamma & 0\\
 0 & 0 & 0 & 0 & 1
\end{pmatrix} \begin{pmatrix}
 0 & 0 & -\alpha' & 0 & -w^0 e^r\\
 0 & 0 & -\alpha' & 0 & 0\\
 -\alpha' & \alpha' & 0 & 0 & 0 \\
 0 & 0 & 0 & 0 & -w^3 \\
 0 & 0 & 0 & 0 & 0
\end{pmatrix} \begin{pmatrix}
 \cosh \gamma & 0 & 0 & \sinh \gamma & 0 \\
 0 & 1 & 0 & 0 & 0 \\
 0 & 0 & 1 & 0 & 0\\
 \sinh \gamma & 0 & 0 & \cosh \gamma & 0\\
 0 & 0 & 0 & 0 & 1
\end{pmatrix}}\\
&={\footnotesize\begin{pmatrix}
 0 & 0 & -\alpha' \cosh \gamma & 0 & -(w^0e^r\cosh \gamma-w^3\sinh\gamma) \\
  0 & 0 & -\alpha' & 0 & 0   \\
 -\alpha' \cosh \gamma & \alpha' & 0 & -\alpha' \sinh\gamma & 0\\
 0 & 0 & \alpha' \sinh\gamma & 0 & (w^0e^r\sinh\gamma-w^3\cosh \gamma)\\
 0 & 0 & 0 & 0 & 0
\end{pmatrix}}
\end{align*}
which followed by a rotation of angle $\beta=\tan^{-1} \sinh \gamma$
around $e_2$ brings the matrix to the following form (note that
$\sin \beta=\tanh \gamma$, $\cos \beta=1/\cosh\gamma$)
\[
{\footnotesize\begin{pmatrix}
 0 & 0 & -\alpha' \cosh \gamma & 0 & -(w^0e^r \cosh \gamma-w^3\sinh\gamma) \\
  0 & 0 & -\alpha' \cosh \gamma & 0 & -\sinh \gamma (w^0e^r\tanh \gamma-w^3)   \\
 -\alpha' \cosh \gamma & \alpha' \cosh \gamma & 0 & 0 & 0\\
 0 & 0 & 0 & 0 & w^0e^r\tanh \gamma-w^3\\
 0 & 0 & 0 & 0 & 0
\end{pmatrix}}.
\]
Choosing $\gamma=\tanh^{-1}(\frac{w^3}{w^0 e^r})$ we arrive at
\[
{\footnotesize\begin{pmatrix}
 0 & 0 & -\alpha' \cosh \gamma & 0 & -w^0e^r/ \cosh \gamma \\
  0 & 0 & -\alpha' \cosh \gamma & 0 & 0   \\
 -\alpha' \cosh \gamma & \alpha' \cosh \gamma & 0 & 0 & 0\\
 0 & 0 & 0 & 0 & 0\\
 0 & 0 & 0 & 0 & 0
\end{pmatrix}}.
\]
Finally, with a boost on the timelike plane $\textrm{Span}(e_0,e_1)$
followed by a translation we obtain
\[
{\footnotesize\begin{pmatrix}
 0 & 0 & -\alpha & 0 & -w^0 \\
  0 & 0 & -\alpha & 0 & 0   \\
 -\alpha & \alpha & 0 & 0 & 0\\
 0 & 0 & 0 & 0 & 0\\
 0 & 0 & 0 & 0 & 0
\end{pmatrix}}.
\]
Thus with a sequence of frame changes we have been able to send
$w^3$ to zero. With a final boost on the timelike plane
$\textrm{Span}(e_0,e_1)$ followed by a translation we send $\alpha$
to $\pm 1$ and $w^0$ to $\epsilon_3 \sqrt{I_3}$ keeping unchanged
all the other matrix entries. A last $\pi$-rotation on the plane
$\textrm{Span}(e_2,e_3)$ sends the possible value $\alpha=-1$ to
$\alpha=1$ giving us the representative 7.

Suppose that $I_2=I_1=0$, $\bar{w}\in \textrm{Ker} F^2$. With a
homogeneous transformation of the reference frame we  send $F$ to
$B$ with $\alpha=1$. At the end of the proof of theorem \ref{gjs} we
have shown that the reference frame can be translated in such a way
that $\bar{w}=(l,l,0,d)$ where $d^2=I_4$. With another translation
we obtain $l=0$ and if $d\le 0$, with a $\pi$-rotation on the plane
$\textrm{Span}(e_2,e_3)$ we obtain $d\ge 0$ and hence
$d=\sqrt{I_4}$, which gives us the representative 8.

\end{proof}

\begin{corollary} \label{pou}
Let $k$ be a Killing field on Minkowski spacetime. Then there is a
reference frame through whose coordinates $k$ takes one of the
twelve forms listed in theorem \ref{vkz}.
\end{corollary}

\begin{proof}
This is just a rephrasing of the previous theorem, given that
$\mathfrak{IL}$ is the Lie algebra of the Killing fields of
Minkowski spacetime.
\end{proof}

\begin{remark}
The closure of the conjugacy class 8 of theorem \ref{vkz} contains
class 10. Indeed, a boost of the frame in the timelike plane
$\textrm{Span}(e_0,e_1)$ shows that
\[
{\footnotesize \begin{pmatrix}
0 & 0  & -\alpha & 0 & 0  \\
0 & 0 & -\alpha & 0 & 0 \\
 -\alpha & \alpha & 0 & 0 & 0 \\
 0 & 0 & 0 & 0 & -\sqrt{I_4} \\
 0 & 0 & 0 & 0 & 0
\end{pmatrix}},
\]
for any $\alpha> 0$ stays in class 8. Taking the limit $\alpha \to
0$ we obtain the representative of class 10. As a consequence, no
continuous $Ad$-invariant function $f:\mathfrak{IL} \to \mathbb{R}$
can distinguish between classes 8 and 10. In particular, no
algebraic Ad-invariant built from the pair $(F^a_{\ b}, w^b)$ can
allow us to distinguish between the two classes.
\end{remark}

\begin{remark}
In a recent paper Barbot considered the conjugacy classes of the
proper orthochronous inhomogeneous Lorentz group \cite[Sect.
6]{barbot05} and obtained a, somewhat coarser, classification. With
respect to that work our proofs are slightly longer because our aim
was to obtain nice representatives by bringing the homogeneous and
translational part into a canonical form. Thanks to our complete set
of Ad-invariants we are able to identify a single conjugacy class
for each choice of allowed Ad-invariants,  and we  are able to tell
exactly which is the conjugacy class of a given transformation by
means of straightforward matrix calculations. On the other hand, the
more geometrical approach by Barbot serves more easily the intuition
for the sake of the classification.

Barbot selects some families of conjugacy classes which, although we
worked on the Lie algebra and he on the Lie group, can be put into
correspondence with our families. The correspondence is as follows.
\begin{description}
\item[{\rm Elliptic}:]  These are our cases 3-6, which correspond to  $I_2=0$, $I_1>0$, and the pure translations 9-11.
\item[{\rm Hyperbolic}:] This is our case 2, which corresponds to $I_2=0$,
$I_1<0$.
\item[{\rm Unipotent}:] These are our cases 7-8, which correspond to
$I_2=I_1=0$, $F\ne 0$, with $I_3>0$ for $7$ and $I_3=0$ for 8.
Barbot's trichotomy is as follows. The {\em linear} case is our case
8 with $I_4=0$. The {\em tangent} case is our case 8 with $I_4\ne
0$. The {\em transverse} case is our case 7.
\item[{\rm Loxodromic}:] This is our case 1 which corresponds to $I_2\ne 0$.
\item[{\rm Parabolic}:] Does not apply in the four dimensional
spacetime case considered here.
\end{description}

\end{remark}

%
%\begin{align*}
%& {\footnotesize \begin{pmatrix}
%0 & 0  & -\alpha & 0 & 0  \\
%0 & 0 & -\alpha & 0 & 0 \\
% -\alpha & \alpha & 0 & 0 & 0 \\
% 0 & 0 & 0 & 0 & -\sqrt{I_4} \\
% 0 & 0 & 0 & 0 & 0
%\end{pmatrix}}
%
%\\
%& ={\footnotesize \begin{pmatrix}
%\frac{\alpha+1/\alpha }{2} & \frac{\alpha-1/\alpha }{2}  & 0 & 0 & 0  \\
%\frac{\alpha-1/\alpha }{2} & \frac{\alpha+1/\alpha }{2} & 0 & 0 & 0 \\
% 0 & 0 & 1 & 0 & 0 \\
% 0 & 0 & 0 & 1 & 0 \\
% 0 & 0 & 0 & 0 & 1
%\end{pmatrix} \begin{pmatrix}
%0 & 0  & -1 & 0 & 0 \\
%0 & 0 & -1 & 0 & 0 \\
% -1 & 1 & 0 & 0 & 0 \\
% 0 & 0 & 0 & 0 & -\sqrt{I_4} \\
% 0 & 0 & 0 & 0 & 0
%\end{pmatrix}\begin{pmatrix}
%\frac{\alpha+1/\alpha }{2} & -\frac{\alpha-1/\alpha }{2}  & 0 & 0 & 0  \\
%-\frac{\alpha-1/\alpha }{2} & \frac{\alpha+1/\alpha }{2} & 0 & 0 & 0 \\
% 0 & 0 & 1 & 0 & 0 \\
% 0 & 0 & 0 & 1 & 0 \\
% 0 & 0 & 0 & 0 & 1
%\end{pmatrix}}
%\end{align*}

\subsection{The Lie wedge}

In section \ref{jnu} we  argued that the semigroup $I\subset
ISO(1,3)^\uparrow$ (resp. $J$) selects those transformations that
are physically reasonable, in the sense that they can be induced by
the dragging of an observer's frame on spacetime.

We would like to select those generators that induce the mentioned
transformation belonging to $I$ (resp. $J$). In other words, we have
to find the Lie algebra counterpart of the semigroup. Fortunately,
there is a well developed Lie theory for subsemigroups of Lie groups
\cite{hilgert89,hilgert93}. If $S$ is a closed subsemigroup of a Lie
group $G$, its Lie wedge (or cone) is the set
\begin{equation} \label{wed}
L(S)=\{X\in \mathfrak{g}: \exp( \mathbb{R^+} X)\subset S\},
\end{equation}
where $\mathbb{R}^+=(0,+\infty)$. The Lie cone is convex because of
the following identity which can be deduced from the
Baker-Campbell-Hausdorff formula \cite[Lemma II.1.1]{hilgert89}
\[
\exp[X+Y]=\lim_{n\to +\infty} [\exp\frac{X}{n}\, \exp\frac{Y}{n}]^n.
\]
The semigroup $J$ is closed, thus the standard theory which can be
found in \cite{hilgert89,hilgert93} applies to it. In particular,
the {\em causal wedge} $L(J)$ is convex.

\begin{remark}
Although $I$ is not closed, we shall define $L(I)$ according to Eq.
(\ref{wed}) and we shall call it the {\em timelike wedge}. The
reader is warned that we are making an abuse of notation, and that
$L(I)$ is not convex.\footnote{There is a definition of Lie wedge
that applies to non-closed semigroup \cite{hilgert89,hilgert93}, but
it would lead back to $L(J)$, while we will need $L(I)$ for our
arguments.}
\end{remark}

%The theory simplifies whenever $S$ is topologically closed, and it
%is worth to note that in our case $J$ is closed. It is interesting
%to observe that t

Let us identify the Lie wedges for the semigroups $I$ and $J$.

\begin{theorem}
The Lie wedges of the semigroups $I$ and $J$ satisfy
\begin{align}
L(J)&=\{ {\footnotesize \begin{pmatrix}
F & - \bar{w} \\
\bar{0}^\intercal & 0
\end{pmatrix}}\!: \ F\in \mathfrak{so}(1,3) \ \textrm{and }  \bar{w} \textrm{ is f.d.\ nonspacelike} \}, \label{kw1}\\
L(J)\backslash L(I)&= \{ {\footnotesize \begin{pmatrix}
F & - \bar{w} \\
\bar{0}^\intercal & 0
\end{pmatrix}}\!: \ F\! \in  \mathfrak{so}(1,3), \  \bar{w} \textrm{ is f.d.\ null and } F\bar{w}\!=\!\lambda \bar{w} \}, \label{kw2}\\
e^{[L(J)\backslash L(I)]} &= \{{\footnotesize \begin{pmatrix}
\Lambda & -\bar{b} \\
\bar{0}^\intercal & 1
\end{pmatrix}}\!: \ \Lambda \!\in SO(1,3), \  \bar{b} \textrm{ is f.d.\ null and } \Lambda\bar{b}=e^\lambda\, \bar{b}
\}, \label{kw3}\\
&\subsetneq J\backslash I, \label{kw4}\\
e^{L(J)}& \subsetneq J, \label{kw5}
\end{align}
where $\lambda\in \mathbb{R}$.

%More in detail, if $\begin{pmatrix}
% F & -  \bar{w} \\
%\bar{0}^\intercal & 0
%\end{pmatrix} \in L(J)\backslash L(I)$ then \[\exp \{\begin{pmatrix}
% F & -  \bar{w} \\
%\bar{0}^\intercal & 0
%\end{pmatrix}\}={\begin{pmatrix}
%\Lambda & -\bar{b} \\
%\bar{0}^\intercal & 1
%\end{pmatrix}},\] where $\Lambda=\exp F$, $\bar{b}=\frac{1}{\lambda}[\exp(\lambda )-1]
%\bar{w}$, and it is understood that $\frac{1}{\lambda}[\exp(\lambda
%)-1]:=1$ for $\lambda=0$.

%
%\[
%\begin{pmatrix}
%F & - \bar{w} \\
%\bar{0}^\intercal & 0
%\end{pmatrix}
%\]
%where $F\in \mathfrak{so}(1,3)^{\uparrow}$ and $\bar{w}$ is f.d.\
%timelike (resp. f.d.\ nonspacelike).
\end{theorem}

\begin{proof}
Let us consider the system which defines the exponential map
(\ref{ndg})-(\ref{ndf}) with initial condition $\Lambda=I$,
$\bar{b}=\bar{0}$.
%\[
%\frac{\dd }{\dd s}
%\begin{pmatrix}
%\Lambda & - \bar{b} \\
%\bar{0}^\intercal & 1
%\end{pmatrix}
%=
%\begin{pmatrix}
%F & - \bar{w} \\
%\bar{0}^\intercal & 0
%\end{pmatrix}
%\begin{pmatrix}
%\Lambda & - \bar{b} \\
%\bar{0}^\intercal & 1
%\end{pmatrix},
%\]
%The matrix equation  is equivalent to the system
%\begin{align}
%\frac{\dd }{\dd s} \,\Lambda &=F\Lambda,  \\
%\frac{\dd }{\dd s} \,\bar{b} &=F\bar{b}+\bar{w}. \label{ndf}
%\end{align}
If $\begin{pmatrix}
\Lambda & - \bar{b} \\
\bar{0}^\intercal & 1
\end{pmatrix}(s)$ belongs to $J$ for all $s>0$ then the same holds
for small positive $s$. By Eq. (\ref{ndf}), since at $s=0$,
$\bar{b}=\bar{0}$, we have that $\bar{w}$ must be f.d.\
nonspacelike.

Conversely, let us suppose that $\bar{w}$ is f.d.\ nonspacelike, and
let $\bar{c}$ be such that $\bar{b}=\Lambda\bar{c}$ (thus $\bar{b}$
is f.d.\ nonspacelike iff $\bar{c}$ is). Eq. (\ref{ndf}) becomes
\begin{equation} \label{mww}
\frac{\dd }{\dd s} \,\bar{c} =\Lambda^{-1}\bar{w}.
\end{equation}
Since the right-hand side is nonspacelike, the integral $\bar{c}$ is
nonspacelike. Equation (\ref{kw1}) is proved.

Let us prove Eq. (\ref{kw2}). Let us suppose that $\bar{b}(s)$ is
f.d.\ nonspacelike for all $s>0$ and lightlike for some
$\tilde{s}>0$. The same holds for $\bar{c}(s)$. We already know that
$\bar{w}$ must be nonspacelike and equation (\ref{mww}) proves that
$\bar{c}(s)$ is a smooth causal curve or
$\bar{c}(s)=\bar{w}=\bar{0}$ for all $s$. Every causal curve which
is not a lightlike pregeodesic  connects chronologically related
points \cite{hawking73}. Thus $\bar{c}(s)$, $0\le s<\tilde{s}$ is a
null pregeodesic curve or $\bar{c}(s)=\bar{0}$. Imposing that the
tangent vector to $\bar{c}(s)$ be proportional to the same null
vector for all $0\le s<\tilde{s}$ gives $\Lambda^{-1}(s)
\bar{w}=f(s) \bar{n} $, for some smooth function $f(s)$. This
equation for $s=0$ gives $\bar{w}=f(0) \bar{n}$ which shows that
$\bar{w}$ is null. Let us differentiate $\bar{w}=f(s) \Lambda(s)
\bar{n} $ and evaluate it at $s=0$. We get $0= f'(0) \bar{n}+ Ff(0)
\bar{n}$ which proves that $w$ is an eigenvector of $F$. Let
$\lambda$ be the eigenvalue, i.e. $F\bar{w}=\lambda \bar{w}$. The
scalar product of Eq. (\ref{ndf}) with $\bar{w}$ gives, $\frac{\dd
}{\dd s} \eta (w,b) =-\lambda \eta(w,b)$, and using the initial
condition $\bar{b}(0)=0$ we obtain $\eta(w,b)=0$. Since, by
assumption, $\bar{b}$ is f.d.\ nonspacelike, we have $\bar{b}=h(s)
\bar{w}$ which plugged back into Eq. (\ref{ndf}) gives $h'=\lambda
h+1$ or $\bar{w}=\bar{0}$. If $\bar{w}\ne \bar{0}$ we infer
$h(s)=\frac{1}{\lambda}[\exp(\lambda s)-1]$ for $\lambda\ne 0$ and
$h(s)=s$ for $\lambda=0$. In summary, if the matrix ${\footnotesize
\begin{pmatrix}
F & - \bar{w} \\
\bar{0}^\intercal & 0
\end{pmatrix}}$ belongs to
$L(J)\backslash L(I)$ then ${\footnotesize\begin{pmatrix}
\Lambda(s) & - \bar{b}(s) \\
\bar{0}^\intercal & 1
\end{pmatrix}}=\exp [s {\footnotesize\begin{pmatrix}
F & - \bar{w} \\
\bar{0}^\intercal & 0
\end{pmatrix}}]$ is such that
$\bar{b}(s)=\frac{1}{\lambda}[\exp(\lambda s)-1]\hat{w}$ where it is
understood that $\frac{1}{\lambda}[\exp(\lambda s)-1]:=s$ for
$\lambda=0$. In particular, $\bar{b}(s)$ is f.d. null for all $s$
and it is an eigenvector for $\Lambda$ with positive eigenvalue
because
\[
\Lambda \bar{b}=(\exp F)\, \bar{b}= (\exp\lambda) \bar{b}.
\]
(Notice that with such a $\bar{b}$ the matrix
${\footnotesize\begin{pmatrix}
\Lambda(s) & - \bar{b}(s) \\
\bar{0}^\intercal & 1
\end{pmatrix}}$ belongs to $J\backslash I$.)
Let us show that conversely every matrix
${\footnotesize\begin{pmatrix}
\Lambda & - \bar{b} \\
\bar{0}^\intercal & 1
\end{pmatrix}}$ belongs to $\exp [L(J)\backslash L(I)]$ provided  $\bar{b}$
is f.d.\ null and it is an eigenvector with positive eigenvalue for
$\Lambda$. The Lie group $SO(1,3)^\uparrow$ is exponential, namely
the exponential map is surjective (for the references see after
theorem \ref{nco}). Thus there is some $F\in \mathfrak{so}(1,3)$
such that $\Lambda =\exp F$. Furthermore, Prop. \ref{nux} shows that
$\Lambda$ and $F$ have the same f.d.\ null eigenvectors thus $F
\bar{b}=\lambda \bar{b}$.  Let us define $\bar{w}=\lambda (\exp
\lambda -1)^{-1} \bar{b}$ for $\lambda\ne 0$, and $\bar{w}=\bar{b}$
for $\lambda=0$, then by the above calculations $\exp
{\footnotesize\begin{pmatrix}
F& - \bar{w} \\
\bar{0}^\intercal & 0
\end{pmatrix}}={\footnotesize\begin{pmatrix}
\Lambda & - \bar{b} \\
\bar{0}^\intercal & 1
\end{pmatrix}}$. We proved Eq. (\ref{kw3}). The fact that the
inclusion (\ref{kw4}) is strict follows taking
${\footnotesize\begin{pmatrix}
\Lambda & - \bar{b} \\
\bar{0}^\intercal & 1
\end{pmatrix}}$ such that $\Lambda\ne I$  (thus some null vectors are not eigenvectors) and $\bar{b}$
is a f.d.\ null vector which is not an eigenvector. This matrix
belongs to $J\backslash I$ but not to $\exp[L(J)\backslash L(I)]$.
The last strict inclusion is an immediate consequence of the
previous one and $\exp I\subset I$.
\end{proof}

As a simple corollary of the previous theorem we obtain

\begin{proposition}
The sets $L(J)$ and $L(J)\backslash L(I)$ are closed and $L(I)$ is
not open. However, $L(I)$ is open in the topology induced on $L(J)$.
\end{proposition}

\subsubsection{The strict inclusion $\exp L(I) \subsetneq I$ and the causal cone of $F$}
%\begin{remark} ()
Suppose that $F\in \mathfrak{so}(1,3)$ is so close to zero that
defined $\Lambda=\exp F \in SO(1,3)^\uparrow$ there is no other
$F'\in \mathfrak{iso}(1,3)$ such that $\Lambda=\exp F'$. We ask the
following question: for which $\bar{b}\in \mathbb{R}^4$ we have
${\footnotesize\begin{pmatrix}
\Lambda & - \bar{b} \\
\bar{0}^\intercal & 1
\end{pmatrix}} \in \exp L(I)$? Is it possible to find some $\bar{b}$
such that this matrix belongs to $I$ but not to $\exp L(J)$?
According to theorem \ref{biq} the $\bar{b}$s which satisfy the
first condition are those which are  causal according to the metric
\[
G=(\frac{F}{e^F-I})^\intercal \eta\, (\frac{F}{e^F-I}),
\]
and f.d.\ timelike according to $\eta$ (recall that $\exp
L(I)\subset I$, then $\bar{w}=\frac{F}{e^F-I} \bar{b}$ cannot be
p.d.\  timelike for otherwise $\bar{b}$ would be  p.d.\ timelike
because of Eq. (\ref{lny})).

%(observe that if $\bar{w}=\frac{F}{e^F-I} \bar{b}$ is p.d.\ (resp.
%f.d.\ ) timelike then $\bar{b}$ is  p.d.\ (resp. f.d.\ ) timelike)
%because of Eq. (\ref{lny})).

 It is instructive to
calculate this metric for the canonical forms $A$ and $B$ of $F$
given by theorem \ref{vkx}. The result is
\begin{align*}
G(A)&={\footnotesize
\begin{pmatrix}
 \frac{-\varphi^2}{2\cosh \varphi -2} & 0  & 0 & 0 \\
 0 & \frac{\varphi^2}{2\cosh \varphi -2} & 0 & 0 \\
 0 & 0 & \frac{\theta^2}{2-2\cos \theta} & 0 \\
 0 & 0 & 0 & \frac{\theta^2}{2-2\cos \theta}
\end{pmatrix}} \\
G(B)&={\footnotesize
\begin{pmatrix}
 -1+\frac{\alpha^2}{12} & -\frac{\alpha^2}{12}  & 0 & 0 \\
 -\frac{\alpha^2}{12}  & 1+\frac{\alpha^2}{12} & 0 & 0 \\
 0 & 0 & 1 & 0 \\
 0 & 0 & 0 & 1
\end{pmatrix}}.
\end{align*}
For the generic $F$ we have $G(F)=L G(A) L^{-1}$ or $G(F)=L G(B)
L^{-1}$, where $L \in SO(1,3)^\uparrow$ and the former or the latter
case apply depending on whether $F$ belongs to the Ad-orbit of $A$
or $B$. Using the inequalities $\frac{\varphi^2}{2\cosh \varphi
-2}\le 1$ (with equality iff $\varphi =0$) and
$\frac{\theta^2}{2-2\cos \theta}\ge 1$ (with equality iff
$\theta=0$), we easily infer that if $F$ is in the Ad-orbit of $A$,
then the causal cone of $G(F)$ is contained inside the causal cone
of $\eta$. Moreover, if $F\ne 0$ it is tangent to it in just two
distinct null directions. As a consequence, the set $I\backslash
\exp L(J)$ is non-empty, it suffices to consider a vector $\bar{b}$
which stay outside the causal cone of $G(F)$ but inside the timelike
cone  of $\eta$. Actually, we can say more, namely that
$\overline{\exp J}\subsetneq J$, because under small perturbations
of $F$ and of $\bar{b}$ as above, $\bar{b}$ keeps staying outside
the causal cone of $G(F)$.

In order to complete our analysis, observe that if $y,x\in
\mathbb{R}$ are such that $y^2\ge x^2$ then
\[
(-1+\frac{\alpha^2}{12})x^2-\frac{\alpha^2}{6}
xy+(1+\frac{\alpha^2}{12}) y^2\ge \frac{\alpha^2}{12} (x-y)^2,
\]
which implies that whenever $F$ belongs to the Ad-orbit of $B$, the
causal cone of $G(F)$ is contained in the  causal cone of $\eta$ and
it is tangent to it in just one null direction. As a consequence, we
can again conclude that the set $I\backslash \exp L(J)$ is
non-empty.

%\end{remark}

We summarize some of these findings through the following
proposition.
\begin{proposition}
We have $\exp L(I)\subsetneq I$ and  $ I\backslash \exp L(J)\ne
\emptyset$. Moreover,  $\overline{\exp J}\subsetneq J$, that is, $J$
is not weakly exponential \cite{hilgert89,hilgert93,hofmann97}.
\end{proposition}

\begin{remark}
One of the consequences of the strict inclusion $\exp L(I)\subsetneq
I$ is that, given two events $p,q\in M$, with $q\in J^{+}(p)$, and
two proper orthochronous bases $\{e_a^p\}$,  $\{e_a^q\}$, at $p$ and
$q$ respectively, it is possible that no observer which moves with
constant acceleration and angular velocity can start with a comoving
base coincident with $\{e_a^p\}$ to later reach $\{e_a^q\}$. One of
the points of this paper is to show that, nevertheless, $\{e_a^p\}$
can be dragged into $\{e_a^q\}$, with the motion of an observer
which moves with constant acceleration and angular velocity.
However, this observer does not necessarily pass through $p$ or $q$.
\end{remark}
%At the infinitesimal level the observer's motion is given, in the
%special frame of the moving observer, by the convex cone of elements
%of type $k(a_i K^i+\omega_k J^k+ H)$, $k>0$.

%\begin{definition}
%The orbits of $\mathfrak{iso}(1,3)$ are called causal if they
%\end{definition}

In the next section we study the physical meaning of these causal
orbits on the Lie algebra.
%
%\begin{proposition}
%Let $\mathcal{P}\in \mathfrak{IL}$ and suppose that for some $s>0$,
%$\exp(\mathcal{P} s)$ sends some point $x$ to its chronological
%future, then $\exp(\mathcal{P} s')$ sends $x$ to its chronological
%future for any $s'>0$.
%
%In other words, if $\exp ({\footnotesize \begin{pmatrix} F& -\bar{w} \\
%\bar{0}^\intercal & 0\end{pmatrix}} s)$ belongs to the semigroup $I$
%for some $s>0$, then ${\footnotesize \begin{pmatrix} F& -\bar{w} \\
%\bar{0}^\intercal & 0\end{pmatrix}} \in L(I)$.
%
%
%
%\end{proposition}

\subsection{The causal orbits}

The Ad action of $ISO(1,3)^{\uparrow}$ on $\mathfrak{iso}(1,3)$ (or
the Ad action of $IL^{\uparrow}_+$ on $\mathfrak{IL}$) generates
orbits which we classified in section \ref{bos}.

It is possible to assign a causal character to these orbits.

\begin{theorem} \label{njx}
The Lie algebra Ad-orbits on $\mathfrak{iso}(1,3)$ which admit some
representative in $L(I)$ belong to the families of orbits (according
to the classification of theorem \ref{vkz}) 1, 2, 4, 7, 9, with
$\epsilon_1=1$ (whenever it applies). Those which admit some
representative in $L(J)\backslash L(I)$ but do not admit any
representative in $L(I)$ belong to the families of orbits 5, 6, 11,
with $\epsilon_1=1$, 8 with $I_4=0$, and the trivial orbit of the
origin (12).
\end{theorem}

\begin{proof}
If we start from representatives 1 or 2 in theorem \ref{vkz}, then,
since $\textrm{Im} F=\textrm{Span}(e_0,e_1)$, with a translation of
the frame we can send $w^0=0$ to $w^0=c$, with $c>0$ arbitrary (in
particular $c>b$ in case 2), and leaving unaltered all the other
matrix entries. After this translation $w^a$ becomes f.d.\ timelike
thus the new representative belongs to $L(I)$. Representatives
4,7,9, with $\epsilon_1=1$ satisfy $w^a$ f.d.\ timelike, thus there
is nothing to prove. As for representatives 5,11, with
$\epsilon_1=1$, 6, or 8 with $I_4=0$, it is clear that $\bar{w}$ is
f.d.\ null  and that $\bar{w}$ is an eigenvector of $F$.

%The orbit of type 5 does not admit any representative

It remains to show that orbits of type 3,  8 with $I_4\ne 0$, 10, do
not have any representative in $L(J)$, that those of type 4, 5, 7,
9, 11, with $\epsilon_1=-1$, do not have any representative in
$L(J)$, and that those of type 5, 11, with $\epsilon_1=1$, 6, 8 with
$I_4=0$, 12, do not have any representative in $L(I)$.

The argument is the same for most of these cases. Any frame change
can be accomplished with a translation  followed by a homogeneous
transformation. In cases 3,  8 with $I_4\ne 0$, 10, and 4, 5,  9,
11, with $\epsilon_1=-1$, $\bar{w}$ is not f.d.\ non-spacelike and
$\textrm{Im} F$ is  a spacelike subspace orthogonal to it (possibly
empty). After the first translation of the frame, the new $\bar{w}$
becomes the sum of the old $\bar{w}$ and of an element belonging to
$\textrm{Im} F$  and hence, is still non f.d.\ non-spacelike.

As for case 8 with $I_4\ne 0$, any frame change can be accomplished
with a translation  followed by a homogeneous transformation. The
former transformation does not change the spacelike causal character
of $\bar{w}$ (since one gets $w^0=w^1$ and possibly $w^3\ne 0$ for
any choice of $\bar{a}$) while the latter preserves its causal
character.

Analogously, in case 7 with $\epsilon_1=-1$, it is easy to check
that operating with a translation to make $w^0$ positive leads to
$w^a$ spacelike.

The proof that classes 5, 11, with $\epsilon_1=1$, 6, 8 with $I_4=0$
and 12, do not have any representative in $L(I)$, proceeds
analogously.

%
%The former transformation does not change the spacelike causal
%character
%
%This fact is pretty clear for case 10 since $\bar{w}$ does not
%change under frame translations and hence preserves its spacelike
%causal character under any frame change.
%
%
%In case 3, 5 or 6 the argument goes as for case 8. $\textrm{Im} F$
%is a spacelike plane thus the first translation sends $\bar{w}$ to a
%spacelike vector.
%
%In cases 11, 12, since $F=0$ the causal character of $\bar{w}$ does
%not change under frame changes thus it is null in any frame.
%
%As for case 8 with $I_4= 0$, any frame change can be accomplished
%with a translation  followed by a homogeneous transformation. The
%former transformation does not change the non-timelike causal
%character of $\bar{w}$ (since one gets $w^0=w^1$  for any choice of
%$\bar{a}$) and the latter preserves its non-timelike character.

\end{proof}

\begin{theorem} \label{kug}
Let $\mathcal{P}\in \mathfrak{IL}$ and suppose that for some $q\in
M$, $(\exp \mathcal{P}) q \in J^{+}(q)$ (resp. $(\exp \mathcal{P}) q
\in I^{+}(q)$), then there is some $q'\in M$ such that $\exp
(\mathcal{P} s)q' \in J^{+}(q')$ (resp. $\exp (\mathcal{P} s)q' \in
I^{+}(q')$) for every $s>0$.

Stated in another way, if an element of $\mathfrak{iso}(1,3)$ has
exponential belonging to $J$ (resp. $I$), then there must be some
representative in its Ad-orbit which belongs to $L(J)$ (resp.
$L(I)$).

\end{theorem}

\begin{proof}
Suppose that  ${\footnotesize
\begin{pmatrix}
F & - \bar{w} \\
\bar{0}^\intercal & 0
\end{pmatrix}}\in \mathfrak{iso}(1,3)$ has exponential belonging to
$J$ (resp. $I$). There is a matrix ${\footnotesize\begin{pmatrix} L & -\bar{a} \\
  \bar{0}^\intercal & 1 \end{pmatrix}} \in ISO(1,3)^\uparrow$ such
  that
  \begin{align*}
  {\footnotesize\begin{pmatrix} F & - \bar{w} \\
\bar{0}^\intercal & 0 \end{pmatrix}}&=
{\footnotesize\begin{pmatrix} L & -\bar{a} \\
  \bar{0}^\intercal & 1 \end{pmatrix}
\begin{pmatrix}
\check{F} & - \check{w} \\
\bar{0}^\intercal & 0
\end{pmatrix}
\begin{pmatrix} L & -\bar{a} \\
  \bar{0}^\intercal & 1 \end{pmatrix}^{-1}} \\
  &={\footnotesize\begin{pmatrix} L & \bar{0} \\
  \bar{0}^\intercal & 1 \end{pmatrix}
  \begin{pmatrix} I & -L^{-1}\bar{a} \\
  \bar{0}^\intercal & 1 \end{pmatrix}
\begin{pmatrix}
\check{F} & - \check{w} \\
\bar{0}^\intercal & 0
\end{pmatrix}
\begin{pmatrix} I & -L^{-1}\bar{a} \\
  \bar{0}^\intercal & 1
  \end{pmatrix}^{-1}
\begin{pmatrix} L & \bar{0} \\
  \bar{0}^\intercal & 1
  \end{pmatrix}^{-1}},
  \end{align*}
where ${\footnotesize\begin{pmatrix}
\check{F} & - \check{w} \\
\bar{0}^\intercal & 0
\end{pmatrix}}$
is one of the representatives of theorem \ref{vkz}. Let
$\bar{c}=L^{-1} \bar{a}$
\begin{align*}
  {\footnotesize\exp \begin{pmatrix} F & - \bar{w} \\
\bar{0}^\intercal & 0 \end{pmatrix}}
  &={\footnotesize\begin{pmatrix} L & \bar{0} \\
  \bar{0}^\intercal & 1 \end{pmatrix}
  \begin{pmatrix} I & -\bar{c} \\
  \bar{0}^\intercal & 1 \end{pmatrix}
(\exp
\begin{pmatrix}
\check{F} & - \check{w} \\
\bar{0}^\intercal & 0
\end{pmatrix})
\begin{pmatrix} I & -\bar{c} \\
  \bar{0}^\intercal & 1
  \end{pmatrix}^{-1}
\begin{pmatrix} L & \bar{0} \\
  \bar{0}^\intercal & 1
  \end{pmatrix}^{-1}.}
  \end{align*}
The frame changes obtained through homogeneous transformations send
$J$ (resp. $I$) to itself, thus the assumption of the theorem is
that
\[
{\footnotesize\begin{pmatrix} I & -\bar{c} \\
  \bar{0}^\intercal & 1 \end{pmatrix}
(\exp
\begin{pmatrix}
\check{F} & - \check{w} \\
\bar{0}^\intercal & 0
\end{pmatrix})
\begin{pmatrix} I & -\bar{c} \\
  \bar{0}^\intercal & 1
  \end{pmatrix}^{-1}},
\]
belongs to $J$ (resp. $I$). Let ${\footnotesize \begin{pmatrix}
\check{\Lambda}(s) & - \check{b}(s) \\
\bar{0}^\intercal & 1
\end{pmatrix}}= \exp ( {\footnotesize
\begin{pmatrix}
\check{F} & - \check{w} \\
\bar{0}^\intercal & 0
\end{pmatrix}}s)$ then we are assuming that
\[
\bar{r}(s):=-(\check{\Lambda}(s)-I) \bar{c}+\check{b}(s),
\]
is f.d.\ nonspacelike (resp. timelike) for some $\bar{c}$ and for
$s=1$. Let us use Eqs. (\ref{ndg})-(\ref{ndf})
\[
\frac{\dd }{\dd s} \,(\bar{r}(s)-\check{w} s)=
\check{F}(\bar{r}(s)-\bar{c}), \qquad \textrm{and }
\bar{r}(0)=\bar{0},
\]
from which we obtain $\bar{r}(1)\in \check{w}+ \textrm{Im}
\check{F}$. This inclusion implies that
${\footnotesize\begin{pmatrix}
\check{F} & - \check{w} \\
\bar{0}^\intercal & 0
\end{pmatrix}}$ belongs to the same orbit of ${\footnotesize\begin{pmatrix}
\check{F} & - \bar{r}(1) \\
\bar{0}^\intercal & 0
\end{pmatrix}}$ (they are connected through a translation of the
frame), which, because of the causal character of $\bar{r}(1)$,
belongs to $L(J)$ (resp. $L(I)$). Thus the Ad-orbit of
${\footnotesize
\begin{pmatrix}
F & - \bar{w} \\
\bar{0}^\intercal & 0
\end{pmatrix}}$ contains an element in $L(J)$ (resp. $L(I)$).

%
%An inspection of theorem \ref{vkz} shows that $\bar{r}(1)$ cannot be
%nonspacelike in cases 3, 4 with $\epsilon_1=-1$, 5  with
%$\epsilon_1=-1$, 6, 8 with $I_4\ne 0$, 9 with $\epsilon_1=-1$, 10,
%11 with $\epsilon_1=-1$. Thus, if $\bar{r}(1)$ is nonspacelike  we
%are in one of the next cases: 1, 2, 4, 5, 7, 9, 11 with
%$\epsilon_1=+1$ (whenever it applies), 8 with $I_4= 0$, and 12. By
%theorem \ref{njx},  ${\footnotesize\begin{pmatrix}
%\check{F} & - \check{w} \\
%\bar{0}^\intercal & 0
%\end{pmatrix}}$ and hence ${\footnotesize\begin{pmatrix}
%{F} & - {w} \\
%\bar{0}^\intercal & 0
%\end{pmatrix}}$ belongs to an Ad-orbit which admits a representative
%in $L(J)$.
%
%Furthermore, an inspection of theorem \ref{vkz} shows that if
%$\bar{r}(1)$ is timelike then we cannot be in case 5, 11 (with
%$\epsilon_1=+1$), 8 with $I_4= 0$, or 12. Thus we must be in one of
%the next cases: 1, 2, 4, 7, 9, with $\epsilon_1=1$. From theorem
%\ref{njx} we conclude that ${\footnotesize\begin{pmatrix}
%\check{F} & - \check{w} \\
%\bar{0}^\intercal & 0
%\end{pmatrix}}$ and hence ${\footnotesize\begin{pmatrix}
%{F} & - {w} \\
%\bar{0}^\intercal & 0
%\end{pmatrix}}$ belongs to an Ad-orbit which admits a representative
%in $L(I)$.

%Let $\bar{q}(s):=\check{\Lambda}^{-1}(s) \bar{r}(s)$ then
%\[
%\frac{\dd }{\dd s} \,\bar{q}(s)=\check{\Lambda}^{-1}(s)
%(-\check{F}\bar{c}+\check{w}).
%\]

\end{proof}

\begin{definition} \label{vkk}
 A conjugacy class of $ISO(1,3)^\uparrow$ is {\em causal} ({\em timelike}) if it admits a
 representative belonging to $J$ (rep. $I$). An Ad-orbit of
 $\mathfrak{iso}(1,3)$ is {\em causal} ({\em timelike})  if it admits an element
 belonging to $L(J)$ (resp. $L(I)$). An Ad-orbit is an {\em
 horismos} Ad-orbit if it is causal but not timelike.
\end{definition}

The logarithm of an element belonging to $ISO(1,3)^\uparrow$ gives
those matrices of $\mathfrak{iso}(1,3)$ whose exponential gives the
original matrix. This set is non-empty because $ISO(1,3)^\uparrow$
is exponential. Clearly, the logarithm sends conjugacy classes into
unions of Ad-orbits.

The previous theorem implies

\begin{corollary}
 The exponential of a causal (timelike) orbit gives a causal (resp.
 timelike) conjugacy class. The logarithm of a causal (resp.
 timelike) conjugacy class is a union of causal (resp.
 timelike) Ad-orbits.
\end{corollary}

We reformulate the relativistic Chasles' theorem emphasizing the
physical content of the classification. For this reason we focus on
the infinitesimal transformations of $M$ whose exponential moves at
least one point $x\in M$ into its causal future $J^{+}(x)$.

In what follows  $\ln I$ ($\ln J$) denotes the subset of
$\mathfrak{iso}(1,3)$ made of matrices whose exponential is
contained in $I$ (resp. $J$). Clearly, $L(I)\subset \ln I$, and
analogously $L(J)\subset \ln J$.

\begin{theorem}[Relativistic Chasles' theorem, causal Lie cone version] \label{vkw}
$\empty$
\begin{itemize}
\item Let  $\mathcal{P}\in \mathfrak{IL}$,
$\mathcal{P}\ne 0$, and suppose that there is a point $q\in M$ such
that $P(s)=\exp (\mathcal{P}s)$ sends $q$  to its timelike future
for some $s>0$. Then it is possible to choose a reference frame such
that $\mathcal{P}$ takes one of the following matrix forms
\begin{align*}
& (a) \  {\footnotesize \begin{pmatrix}
 0 & -a  & 0 & 0 & -1\\
 -a& 0 & 0 & 0 & 0\\
 0 & 0 & 0 & \omega & 0\\
 0 & 0 & -\omega & 0 & 0 \\
0 & 0 & 0 & 0 & 0
\end{pmatrix} } \tau=(a K^1+\omega J^1+ H)\tau, \quad \parbox{3.5cm}{where $a > 0$, $\omega \ne 0$,} \\
 & (b) \
 {\footnotesize\begin{pmatrix}
 0 & 0 & -a & 0 & -1\\
 0 & 0 & -\omega & 0 & 0\\
 -a & \omega & 0 & 0 & 0 \\
 0 & 0 & 0 & 0 & 0 \\
 0 & 0 & 0 & 0 & 0
\end{pmatrix} } \tau=(a K^2-\omega J^3+ H)\tau, \quad \parbox{4cm}{where $a,\omega\ge 0$.}
\end{align*}
where $\tau>0$. Stated in another way, the orbits of
$\mathfrak{iso}(1,3)$ under the Ad action of $ISO(1,3)^\uparrow$
which admit an element in $\ln I$ admit a representative which is
either of type (a) (if $I_2\ne 0$) or of type (b) (if $I_2=0$). The
constants $a,\omega,\tau$, are arbitrary as long as they satisfy
%are related to the Ad-invariants through the constraints
\begin{align}
& & (a^2-\omega^2)\tau^2&=-2  I_1, \label{cf1}\\
I_2\ne 0 \ &\Rightarrow&  a\omega \tau^2&=I_2,\\
I_2=0 \ &\Rightarrow& \ \omega^2 \tau^4&= I_3, \\
\qquad I_1=I_2=I_3=0 \ (F=0) \ &\Rightarrow& \  \tau^2&=- I_4.
\qquad \qquad  \label{cf3}
\end{align}
%where the latter equation applies only if $I_2=0$.
%Furthermore, if
%$I_1\ge 0$, $I_2=0$, $I_3>0$, then $\omega=\alpha/\sqrt{I_3}$ where
%$\alpha>0$ is such that $ I_1 \alpha=2 I_1^2$ (thus $\alpha$ is
%arbitrary for $I_1=0$).

Whenever case (b) applies, it is possible to choose the frame in
such a way that $0\le \omega\le a$ (if $I_1\le 0$) or $a=0$,
$\omega=2I_1/\sqrt{I_3}$ (if $I_1>0$). We have pure rotation if  $I_2=0$, $I_1>0$ or $I_1=I_2=I_3=0$  ($F=0$).
If pure rotation  does not apply, then  $\tau>0$ can be  chosen arbitrarily, and once this is done, $\vert \omega\vert$ and $a$ are uniquely determined.

%In the former  case  the constants $a,\omega,$ are uniquely
%determined by the orbit, while in the latter case $a/\omega$ is
%uniquely determined but $a$ can change changing frame.

%one and only one of the representatives given above (apart for the
%mentioned freedom in $\alpha$) (the trivial orbit of the origin
%contains only the zero matrix).

\item Let  $\mathcal{P}\in \mathfrak{IL}$,  and suppose that there is a point $q\in M$ such
that $P(s)=\exp (\mathcal{P}s)$ sends $q$  to some point in
$J^{+}(q)\backslash \{q\}$ for some $s>0$, and that $\mathcal{P}$
does not have the property of the previous point. Then it is
possible to choose a reference frame such that $\mathcal{P}$ takes
one of the following matrix forms
\begin{align*}
(c) & \ {\footnotesize
\begin{pmatrix}
0 & 0  & 0 & 0 & -1  \\
0 & 0 & 0 & 0 & -\epsilon_2 \\
 0 & 0 & 0 & \omega & 0 \\
 0 & 0 & -\omega & 0 & 0 \\
 0 & 0 & 0 & 0 & 0
\end{pmatrix}} \lambda=(\omega J^1-\epsilon_2 P^1+H)\lambda \quad \parbox{4cm}{where $\omega \ge 0$, $\lambda>0$, \\ and $\epsilon_2=\pm
1$,}
\end{align*}
and where $\lambda$ and $\omega$ are arbitrary as long as they
satisfy $\lambda \omega=\sqrt{2I_1}$, or
\begin{align*}
(d) & \ {\footnotesize
\begin{pmatrix}
0 & 0  & -\eta & 0 & -1  \\
0 & 0 & -\eta & 0 & -1 \\
 -\eta & \eta & 0 & 0 & 0 \\
 0 & 0 & 0 & 0 & 0 \\
 0 & 0 & 0 & 0 & 0
\end{pmatrix}} \lambda=[( K^2- J^3)\eta -P^1+H]\lambda \quad \parbox{4cm}{where $\eta \ne 0$, \\ and $\lambda>0$,}
\end{align*}
and where $\eta$ and $\lambda$ are arbitrary as long as they satisfy
the constraints.

Stated in another way, the orbits of $\mathfrak{iso}(1,3)$ under the
Ad action of $ISO(1,3)^\uparrow$ which admit an element in $\ln J$
but none in $\ln I$, admit  representative (c) (if $I_1>0$) or (d)
(if $I_1=0$).
%The constants $\theta$ and $\epsilon_2$ are uniquely determined by
%the orbit.
\end{itemize}

\end{theorem}

\begin{proof}
According to theorem \ref{kug} we can suppose that $q$ is sent to
its (timelike) causal future for every $s>0$.

Let us choose a reference frame with origin at $q$ and let ${\footnotesize \begin{pmatrix} F& -\bar{w} \\
\bar{0}^\intercal & 0\end{pmatrix}}$ be the corresponding matrix
form of $\mathcal{P}$. Let ${\footnotesize \begin{pmatrix} \Lambda(s)& -\bar{b}(s) \\
\bar{0}^\intercal & 1\end{pmatrix}}$ be the matrix of
$\exp(\mathcal{P}s)$. Since
$\bar{b}(s)$ is timelike for every $s>0$ we have ${\footnotesize \begin{pmatrix} F& -\bar{w} \\
\bar{0}^\intercal & 0\end{pmatrix}} \in L(I)$. By theorem \ref{njx}
the frame can actually be chosen in such a way that ${\footnotesize \begin{pmatrix} F& -\bar{w} \\
\bar{0}^\intercal & 0\end{pmatrix}}$ takes one of the forms
1,2,4,7,9, (with $\epsilon_1=1$) of theorem \ref{vkz}. We have  to
show that in each of these cases, through a suitable frame change,
we can bring the matrix to forms (a) or (b).

We are going to show that we can obtain (a) starting from 1,   and
(b) from 2, 4, 7 or 9. In other words we get (a) if $I_2\ne 0$ and
(b) if $I_2=0$.

%Clearly, since $I_2$ is an Ad-invariant, we can obtain (a) only
%starting from 1, and we can obtain (b) only starting from 2,4,7 or
%9.

Thus let us suppose that ${\footnotesize \begin{pmatrix} F& -\bar{w} \\
\bar{0}^\intercal & 0\end{pmatrix}}$ is the representative given in
point 1,   theorem \ref{vkz}. Let us observe that $\varphi,\theta\ne
0$. Since $\textrm{Im} F\supset \textrm{Span}(e_0,e_1)$, through a
translation of the frame we reach the matrix form
\[
{\footnotesize \begin{pmatrix}
 0 & -\varphi & 0 & 0 & -\tau \\
 -\varphi & 0 & 0 & 0 & 0 \\
 0 & 0 & 0 & \theta & 0\\
 0 & 0 & -\theta & 0 & 0\\
 0 & 0 & 0 & 0 & 0
\end{pmatrix}}={\footnotesize\begin{pmatrix}
 0 & -a & 0 & 0 & -1 \\
 -a & 0 & 0 & 0 & 0 \\
 0 & 0 & 0 & \omega & 0\\
 0 & 0 & -\omega & 0 & 0\\
 0 & 0 & 0 & 0 & 0
\end{pmatrix}}\tau,
\]
where $\tau>0$ can be chosen arbitrarily and $a=\varphi/\tau$,
$\omega=\theta/\tau$.

Let us come to the cases that will lead us to the form (b).

%Suppose that ${\footnotesize \begin{pmatrix} F& -\bar{w} \\
%\bar{0}^\intercal & 0\end{pmatrix}}$ is the representative given in
%point 2, theorem \ref{vkz}, with $\epsilon_1=1$ and $b=0$. Since
%$\textrm{Im} F=\textrm{Span}(e_0,e_1)$, through translation of the
%frame we reach the matrix form

In case 9 set $\tau=\sqrt{-I_4}$, $a=\omega=0$.

Suppose that we are in case 2. Through translation of the frame we
reach the matrix form
\[
{\footnotesize \begin{pmatrix}
 0 & -\varphi & 0 & 0 & -c \\
 -\varphi & 0 & 0 & 0 & 0 \\
 0 & 0 & 0 & 0 & 0\\
 0 & 0 & 0 & 0 & -b\\
 0 & 0 & 0 & 0 & 0
\end{pmatrix}},
\]
where we can choose $c>\vert b\vert$. The next identity holds
\begin{align*}
  & {\footnotesize\begin{pmatrix}
 \cosh \gamma & 0 & 0 & -\sinh \gamma & 0 \\
 0 & 1 & 0 & 0 & 0 \\
 0 & 0 & 1 & 0 & 0\\
 -\sinh \gamma & 0 & 0 & \cosh \gamma & 0\\
 0 & 0 & 0 & 0 & 1
\end{pmatrix} \begin{pmatrix}
 0 & -\varphi & 0 & 0 & -c \\
 -\varphi & 0 & 0 & 0 & 0 \\
 0 & 0 & 0 & 0 & 0\\
 0 & 0 & 0 & 0 & -b\\
 0 & 0 & 0 & 0 & 0
\end{pmatrix} \begin{pmatrix}
 \cosh \gamma & 0 & 0 & \sinh \gamma & 0 \\
 0 & 1 & 0 & 0 & 0 \\
 0 & 0 & 1 & 0 & 0\\
 \sinh \gamma & 0 & 0 & \cosh \gamma & 0\\
 0 & 0 & 0 & 0 & 1
\end{pmatrix}}\\
&={\footnotesize\begin{pmatrix}
 0 & -\varphi \cosh \gamma & 0 & 0 & -(c\cosh \gamma-b\sinh\gamma) \\
 -\varphi \cosh \gamma & 0 & 0 & -\varphi \sinh\gamma & 0 \\
 0 & 0 & 0 & 0 & 0\\
 0 & \varphi\sinh\gamma & 0 & 0 & (c\sinh\gamma-b\cosh \gamma)\\
 0 & 0 & 0 & 0 & 0
\end{pmatrix}},
\end{align*}
Let $\tau>0$ be defined by $\tau:=\sqrt{c^2-b^2}$. The freedom in
$c$ shows that $\tau>0$ can be chosen arbitrarily. Since $c,\tau
>0$ we can choose $\gamma$ such that $\tanh \gamma= b/c$, so that
$c\sinh\gamma-b\cosh \gamma=0$ and $c\cosh
\gamma-b\sinh\gamma=\tau$. Thus defining $a= \varphi c/\tau^2>0$ and
$\omega= \varphi b/\tau^2$ we obtain (observe that
$c=\tau/\sqrt{1-(\omega/a)^2}$ and $\frac{b}{c}=\frac{\omega}{a}$)
\[
{\footnotesize\begin{pmatrix}
 0 & -a & 0 & 0 & -1 \\
 -a & 0 & 0 & -\omega & 0 \\
 0 & 0 & 0 & 0 & 0\\
 0 & \omega & 0 & 0 & 0\\
 0 & 0 & 0 & 0 & 0
\end{pmatrix}\tau},
\]
which through a suitable rotation of the reference frame can be
brought to the form
\[
{\footnotesize\begin{pmatrix}
 0 & 0 & -a & 0 & -1 \\
 0 & 0 & -\omega & 0 & 0 \\
 -a & \omega & 0 & 0 & 0\\
 0 & 0 & 0 & 0 & 0\\
 0 & 0 & 0 & 0 & 0
\end{pmatrix}\tau}.
\]
We observe that in this case $0\le \omega <a$. For future reference
we record that the original matrix can be rewritten
\begin{equation} \label{bat}
{\footnotesize \begin{pmatrix}
 0 & -\varphi & 0 & 0 & 0 \\
 -\varphi & 0 & 0 & 0 & 0 \\
 0 & 0 & 0 & 0 & 0\\
 0 & 0 & 0 & 0 & -b\\
 0 & 0 & 0 & 0 & 0
\end{pmatrix}}={\footnotesize \begin{pmatrix}
 0 & -a /\gamma & 0 & 0 & 0\\
 -a /\gamma & 0 & 0 & 0 & 0 \\
 0 & 0 & 0 & 0 & 0\\
 0 & 0 & 0 & 0 & -{(\omega/a)}\gamma\\
 0 & 0 & 0 & 0 & 0
\end{pmatrix}} \tau,
\end{equation}
where $\gamma:=1/\sqrt{1-(\omega/a)^2}$.

In case 7 we first boost the frame in the plane
$\textrm{Span}(e_0,e_1)$ and make a translation so that $\alpha=1$
gets replaced by any chosen $\alpha>0$ and the entry $-\sqrt{I_3}$
gets replaced by $-\sqrt{I_3}/\alpha$. We define
$a=\omega=\alpha^2/\sqrt{I_3}$, and $\tau=\sqrt{I_3}/\alpha$. We
observe that the common module of $a$ and $\omega$ can be chosen
freely due to the freedom in $\alpha$.

Suppose that ${\footnotesize \begin{pmatrix} F& -\bar{w} \\
\bar{0}^\intercal & 0\end{pmatrix}}$ is the representative given in
point 4, theorem \ref{vkz}, with $\epsilon_1=1$. Let us observe that
$\varphi=0$; $\theta,b \ne 0$. Defined $\tau=b$, $a=0$ and
$\omega=\theta/b=2I_1/\sqrt{I_3}$, after a rotation of the frame we
obtain the matrix form (b) with $a=0$, $\omega\ne 0$.

So far all the cases that we have considered that lead to case (b)
with $a,\omega \ne 0$ show that we can always satisfy the inequality
$0\le \omega \le a$. Of all the cases that we have considered just
case 4 gives $a<\omega$, but we can regard it as a case of aligned
angular velocity and acceleration. In any case, it is convenient to
observe that case 4 can be brought to the form (b) with  $a,\omega$
such that $0<a<\omega$, through a sequence of translation, boost and
rotation following calculations similar to those of case 2.

The statement concerning Eqs. (\ref{cf1})-(\ref{cf3}) can be easily
checked calculating the invariants for (a) and (b).

The last point is an easy consequence of theorem \ref{njx}, through
inspection of cases 5,11 with $\epsilon_1=1$, and 8 with $I_4=0$, of
theorem \ref{vkz}.

%expressing in each case, $b,\varphi,\theta$, in terms of the
%invariant.
\end{proof}

\subsection{Lorentzian extension of Chasles' theorem}

We are ready to prove that any orientation and time orientation
preserving isometry of Minkowski spacetime which sends some point to
its chronological future, can be accomplished through the dragging
of spacetime points by the motion of an observer's reference frame,
where the observer moves with constant acceleration and angular
velocity for some proper time interval.

\begin{theorem} \label{vtw}(Relativistic Chasles' theorem, group version, timelike
part)  Suppose that $P:M\to M$, $P\in IL^{\uparrow}_+$, sends some
point to its chronological future, then there is a reference frame
on $M$ with respect to whose coordinates  $P$ takes one of the
following matrix forms
\begin{align*}
 (a) \qquad & \exp[(a K^1+\omega J^1+H)\tau]= \exp({\footnotesize \begin{pmatrix}
 0 & -a  & 0 & 0 & -1\\
 -a& 0 & 0 & 0 & 0\\
 0 & 0 & 0 & \omega & 0\\
 0 & 0 & -\omega & 0 & 0 \\
0 & 0 & 0 & 0 & 0
\end{pmatrix} } \tau) \\
&={\footnotesize\begin{pmatrix}
 \cosh (a\tau) & -\sinh (a\tau)  & 0 & 0 & -\frac{1}{a} \sinh (a\tau)  \\
 -\sinh (a\tau) & \cosh (a\tau) & 0 & 0 & \frac{1}{a} [\cosh (a\tau) -1] \\
 0 & 0 & \cos (\omega\tau) &  \sin (\omega\tau)  & 0\\
 0 & 0 & -\sin (\omega\tau) & \cos (\omega\tau) & 0\\
 0 & 0 & 0 & 0 & 1
\end{pmatrix}}, \\
(b)\qquad  & \exp[(a K^2-\omega J^3+ H)\tau]= \exp(
{\footnotesize\begin{pmatrix}
 0 & 0 & -a & 0 & -1\\
 0 & 0 & -\omega & 0 & 0\\
 -a & \omega & 0 & 0 & 0 \\
 0 & 0 & 0 & 0 & 0 \\
 0 & 0 & 0 & 0 & 0
\end{pmatrix} } \tau)\\
&={\footnotesize \begin{pmatrix}
 1+ (a \tau)^2/2 &  -(a \tau)^2/2 & -a \tau & 0 & -\tau-a^2 \tau^3/6 \\
 (a \tau)^2/2 & 1-(a \tau) ^2/2 & -a \tau  & 0 & -a^2 \tau^3/6\\
 -a \tau & a \tau & 1 & 0 & a \tau^2/2\\
 0 & 0 & 0 & 1 & 0\\
 0 & 0 & 0 & 0 & 1
\end{pmatrix}}, \\
(c) \qquad  & \exp[(aK^2-\omega J^3+H)\tau]={\footnotesize\exp
(\begin{pmatrix}
 0 & 0 & -a & 0 & -1 \\
 0 & 0 & -\omega & 0 & 0 \\
 -a & \omega & 0 & 0 & 0\\
 0 & 0 & 0 & 0 & 0\\
 0 & 0 & 0 & 0 & 0
\end{pmatrix}\tau)}\\
& = {\footnotesize \left(\begin{matrix}
 1+\gamma^2[\cosh (a\tau/\gamma )-1] & -{(\omega/a)\gamma^2} [\cosh ({a\tau }/\gamma)-1]   \\
(\omega/a) \gamma^2  [\cosh ({a\tau }/\gamma)-1] & 1-{(\omega/a)^2}\gamma^2  [\cosh ({a\tau }/\gamma)-1] \\
 - \gamma \sinh ({a\tau }/\gamma) & {(\omega/a) }\gamma \sinh ({a\tau }/\gamma) \\
 0 & 0  \\
 0 & 0
\end{matrix}\right. \cdots }\\
&\qquad \ {\footnotesize \cdots \left.
\begin{matrix}
 - \gamma \sinh ({a\tau }/\gamma) & 0 &   {(\omega/a)^2 \gamma^2 \tau}-\frac{1}{a} \gamma^{3} \sinh ({a\tau }/\gamma) \\
- {(\omega/a) }\gamma \sinh ({a\tau }/\gamma) & 0 & {(\omega/a)\gamma^2\tau } -{(\omega/a^2) } \gamma^3 \sinh ({a\tau }/\gamma)\\
  \cosh ({a\tau }/\gamma) & 0 & \frac{1}{a} \gamma^2 [\cosh ({a\tau }/\gamma)-1]\\
 0 & 1 & 0\\
 0 & 0 & 1
\end{matrix}\right) },
\end{align*}
where $\gamma(a,\omega):=1/\sqrt{1-(\omega/a)^2}$, $\tau>0$ and,
furthermore, in (a) $a\ge 0$, in (b) $a=\omega\ne 0$, in (c) $0\le
\omega < a$. The arbitrariness in $a$, $\omega$, $\tau$, is the same
as that given in theorem \ref{vkw}.
%According to the invariants of $|la (a), (b) or (c) apply respectively if $I_2\ne 0$,

\end{theorem}

\begin{proof}
By theorem \ref{biq}  $ISO(1,3)^\uparrow$ is exponential, thus there
is some $\mathcal{P}\in \mathfrak{IL}$ such that $P=\exp
\mathcal{P}$. The remainder of the theorem follows from theorem
\ref{vkw} after some algebra (the last matrix can also be obtained
through a transformation of the frame from Eq. (\ref{bat})).
\end{proof}

The previous theorem involves the exponential of elements of
$\mathfrak{iso}(1,3)$. The reader interested in general closed
exponentiation formulas is referred to
\cite{zeni90,leite99,fredsted01}.

\begin{theorem} \label{vhs} (Relativistic Chasles' theorem, group version, horismos part)
Suppose that $P:M\to M$, $P\in IL^{\uparrow}_+$, sends some point
$q\in M$ to some point in $J^{+}(q)\backslash\{q\}$, but none to its
chronological future, then there is a reference frame on $M$ with
respect to whose coordinates $P$ takes one of the following matrix
forms
\begin{align*}
(a) \qquad & \exp[(\omega J^1-\epsilon_2 P^1+H)\lambda]= \exp
({\footnotesize
\begin{pmatrix}
0 & 0  & 0 & 0 & -1  \\
0 & 0 & 0 & 0 & -\epsilon_2 \\
 0 & 0 & 0 & \omega & 0 \\
 0 & 0 & -\omega & 0 & 0 \\
 0 & 0 & 0 & 0 & 0
\end{pmatrix}} \lambda ) \\
&={\footnotesize\begin{pmatrix}
 1 & 0  & 0 & 0 & -\lambda  \\
 0 & 1 & 0 & 0 & -\epsilon_2 \lambda \\
 0 & 0 & \cos (\omega \lambda) &  \sin (\omega \lambda) & 0\\
 0 & 0 & -\sin (\omega \lambda) & \cos (\omega \lambda) & 0\\
 0 & 0 & 0 & 0 & 1
\end{pmatrix}}, \\
(b) \qquad & \exp\{[( K^2- J^3)\eta -P^1+H]\lambda\}= \exp(
{\footnotesize\begin{pmatrix}
0 & 0  & -\eta & 0 & -1  \\
0 & 0 & -\eta & 0 & -1 \\
 -\eta & \eta & 0 & 0 & 0 \\
 0 & 0 & 0 & 0 & 0 \\
 0 & 0 & 0 & 0 & 0
\end{pmatrix}} \lambda)\\
&={\footnotesize \begin{pmatrix}
 1+ (\eta \lambda)^2/2 &  -(\eta \lambda)^2/2 & -\eta \lambda & 0 & -\lambda \\
 (\eta \lambda)^2/2 & 1-(\eta \lambda) ^2/2 & -\eta \lambda  & 0 & -\lambda \\
 -\eta \lambda & \eta \lambda & 1 & 0 & 0\\
 0 & 0 & 0 & 1 & 0\\
 0 & 0 & 0 & 0 & 1
\end{pmatrix}}
\end{align*}
where $\lambda\ge 0$ and, moreover, in (a) $\omega > 0$,
$\epsilon_2=\pm 1$, $\lambda \omega = \sqrt{2I_1}$. The
arbitrariness in $\lambda$, $\omega$, $\eta$, is the same as that
given in theorem \ref{vkw}.
\end{theorem}

\begin{proof}
By theorem \ref{biq}  $ISO(1,3)^\uparrow$ is exponential, thus there
is some $\mathcal{P}\in \mathfrak{IL}$ such that $P=\exp
\mathcal{P}$. The remainder of the theorem follows from theorem
\ref{vkw} after some algebra.
\end{proof}

The transformations of type (a) might be called {\em lightlike
screws}. We mention that in \cite{synge56,shaw69} the term {\em
screw} is used for what we call {\em roto-boost}. Since {\em
roto-boosts} appear already in the study of the Lorentz group, which
does not include translations, it seems to be inappropriate to use
the term {\em screw} for those transformations.

%In a sense the generalization of Chasles' theorem is more
%interesting than that of Euler's theorem as it might become the
%starting point for a relativistic screw theory.

\begin{remark}
With reference to the canonical motions (a), (b) and (c) of theorem
\ref{vtw}, it is interesting to calculate the position
$\Lambda^{-1}\bar{b}$ of the frame with respect to its coordinates
at time $\tau=0$. They are
\begin{align*}
(a) & \quad {\footnotesize \begin{pmatrix} \frac{1}{a} \sinh (a
\tau)
\\ \frac{1}{a} [\cosh (a \tau) -1] \\ 0 \\ 0
\end{pmatrix}}, \qquad
(b)    \quad {\footnotesize \begin{pmatrix} \tau+a^2 \tau^3/6
\\ a^2\tau^3/6 \\ a \tau^2/2 \\ 0
\end{pmatrix}},
\end{align*}
where we omit the expression for (c) which is complex and not
particularly illuminating. It seems curious that we get a rather
simple polynomial expression for case (b) which corresponds to equal
and orthogonal acceleration and angular velocity.
\end{remark}

We end this work giving in table \ref{table4b} and \ref{table4c} the
classification of timelike and horismos Ad-orbits of
$\mathfrak{iso}(1,3)$. There we choose the simplest representative
which, however, might not belong to $L(I)$ (resp. $L(J)\backslash
L(I)$). Nevertheless, we keep the parametrization as it is inherited
by its conjugacy equivalent which belongs to $L(I)$ (resp.
$L(J)\backslash L(I)$). The last column reminds us that once the
parameters selecting the orbit have been fixed, the freedom left in
the choice of simplifying reference frame selects some
characteristic geometric object. These ingredients provide the
generalization to the relativistic case of Mozzi and Chasles'
instantaneous axis of rotation.

\section{Conclusions}

We have generalized Chasles' theorem to the Lorentzian spacetime
case, proving that every inhomogeneous proper orthochronous Lorentz
transformation, which sends some point to its chronological future,
can be obtained through the displacement of an observer which moves
at constant angular velocity and constant acceleration (theorems
\ref{vkw} and \ref{vtw}). We have also given an horismos version of
this result in which a lightlike geodesic plays the role of the
observer's worldline (theorem \ref{vhs}).

Intuitively, this result states that if the isometry satisfies the
mentioned causality requirement, then it is generated through some
canonical frame motion along the natural causal entities that live
on spacetime: observers and light rays.

%Thus, It implies that we do not lose generality in studying Lorentz
%transformations just in the restrictive framework in which they
%arise  motion of massive and massless particles.

In order to accomplish this result we first proved the
exponentiality of the proper orthochronous inhomogeneous Lorentz
group (Theor. \ref{biq}). We studied the Lie algebra introducing a
complete set of Ad-invariants (Theor. \ref{gjs}) which allowed us to
classify the Ad-orbits (Theor. \ref{vkz}). As a corollary, we
obtained a classification of the adjoint inequivalent Killing fields
of Minkowski spacetime (Theor. \ref{vkz}, Cor. \ref{pou}).

It is clear that space translations, while being isometries, are not
generated by any observer's causally meaningful motion. In order to
obtain a relativistic version of Chasles' theorem it was necessary
to impose some causality condition. The weakest is the requirement
that the transformation sends some point to its chronological
(causal) future. Keeping this observation in mind we went to study
the causal semigroup of the inhomogeneous Lorentz group and its Lie
cone. In this respect, we connected this weak causality condition
with the apparently stronger condition which wants the logarithm of
the transformation on the Lie wedge \ref{kug}, and we identified
those Ad-orbits that admit a causal representative (Theor.
\ref{njx}). Finally, we proved the relativistic generalization of
Chasles' theorem.

In our analysis we  payed special attention to the geometrical
content of the Lorentz transformations, summarizing the
possibilities in tables \ref{table4b} and \ref{table4c}. Given the
conjugacy class (or Ad-orbit) and the appropriate geometric
information, it is then possible to fully recover the transformation
and, more importantly, to grasp its physical meaning.

%\section*{Acknowledgments}
%This work has been partially supported by GNFM of INDAM.

\clearpage

\begin{table}[ht]
\caption{Relativistic Chasles' theorem  and reconstruction (timelike
Lie wedge version).  The simplest representatives here displayed are
not necessarily those belonging to $L(I)$, nevertheless they are
parametrized keeping in mind the physical interpretation of their
equivalents which  belong to $L(I)$. }
 {\small
$\!\!\!\!\!\!\!\!\!\!\!\!\!\!\!$\begin{tabular}{llccc}
\\
\hline
Type & \parbox{4cm}{\begin{center}Families of timelike orbits (Def. \ref{vkk})\\
 (some matrices are given up to a positive factor) \end{center}} &  \parbox{2cm}{Parameters \\ (omitted positive factor)} & Description & \parbox{1.7cm}{Geometric ingredients} \\
\hline
\\
(p1) & {\footnotesize $\begin{pmatrix}
 0 & 0 & 0 & 0 & -1 \\
 0 & 0 & 0 & 0 & 0 \\
 0 & 0 & 0 & 0 & 0\\
 0 & 0 & 0 & 0 & 0\\
 0 & 0 & 0 & 0 & 0
\end{pmatrix}$} & \parbox{1.5cm}{$[none]$}
  & \parbox{2.8cm}{ timelike translation (inertial motion)} & \parbox{1.7 cm}{timelike direction} \\
\\
(p2) & {\footnotesize $\begin{pmatrix}
 0 & - a  & 0 & 0 & 0 \\
 - a & 0 & 0 & 0 & 0 \\
 0 & 0 & 0 &   \omega  & 0\\
 0 & 0 & - \omega & 0 & 0\\
 0 & 0 & 0 & 0 & 0
\end{pmatrix}$} & \parbox{1.8cm}{$a >0$, \\ $\omega \ne 0$.}
  & \parbox{2.8cm}{acceleration aligned   with angular velocity}  & \parbox{1.7 cm}{oriented timelike 2-plane with origin} \\
  \\
(p3) & {\footnotesize $\begin{pmatrix}
 0 & - a  & 0 & 0 & 0  \\
 - a & 0 & 0 & 0 & 0 \\
 0 & 0 & 0 &  0  & 0\\
 0 & 0 & 0 & 0 & 0\\
 0 & 0 & 0 & 0 & 0
\end{pmatrix}$} & \parbox{1.8cm}{$a >0$.}
  & acceleration & \parbox{1.7 cm}{oriented spacelike 2-plane} \\
\\
 (p4) & {\footnotesize $\begin{pmatrix}
 0 & 0  & 0 & 0 & - 1 \\
 0  & 0 & 0 & 0 & 0 \\
 0 & 0 & 0 &  \omega & 0\\
 0 & 0 & - \omega & 0 & 0 \\
 0 & 0 & 0 & 0 & 0
\end{pmatrix}$} & \parbox{1.5cm}{ $\omega>0$,}
  & \parbox{1.5cm}{ rotation} & \parbox{1.7 cm}{oriented timelike 2-plane and timelike direction on it} \\
  \\
(p5) & {\footnotesize $\begin{pmatrix}
 0 & -a /\gamma & 0 & 0 & 0 \\
 -a /\gamma & 0 & 0 & 0 & 0 \\
 0 & 0 & 0 & 0 & 0\\
 0 & 0 & 0 & 0 & -{(\omega/a)}\gamma\\
 0 & 0 & 0 & 0 & 0
\end{pmatrix}$} & \parbox{3cm}{ \mbox{$\gamma:=1/\sqrt{1-(\omega/a)^2}$} \\$a >0$, \\ $0< \omega<a$. }
  & \parbox{2.8cm}{the acceleration and angular velocity are orthogonal} & \parbox{1.7 cm}{oriented spacelike 2-plane and oriented spacelike direction on it} \\
 \\
 (p6) & {\footnotesize $\begin{pmatrix}
 0 &  0 & -a & 0 & -1 \\
 0 & 0 & -\omega  & 0 & 0\\
 -a & \omega & 0 & 0 & 0\\
 0 & 0 & 0 & 0 & 0\\
 0 & 0 & 0 & 0 & 0
\end{pmatrix}$}
  &  \parbox{2cm}{$a=\omega> 0$. }  & \parbox{2.8cm}{the acceleration and angular velocity are orthogonal}
  & \parbox{1.8cm}{oriented lightlike 2-plane}  \\
 \\
 \hline
\end{tabular} }
\label{table4b}
\end{table}

\begin{table}[ht]
\caption{Relativistic Chasles' theorem  and reconstruction (horismos
Lie wedge version)} {\small
$\!\!\!\!\!\!\!\!\!\!\!\!\!\!\!$\begin{tabular}{llccc}
\\
\hline
Type & \parbox{4cm}{\begin{center}Families of horismos orbits (Def. \ref{vkk})\\
 \end{center}} & Parameters & Description & \parbox{1.7cm}{Geometric ingredients} \\
\hline \\
(p7) & {\footnotesize $\begin{pmatrix}
 0 & 0  & 0 & 0 & -1 \\
 0  & 0 & 0 & 0 & -\epsilon_2 \\
 0 & 0 & 0 &   1 & 0\\
 0 & 0 & - 1 & 0 & 0 \\
 0 & 0 & 0 & 0 & 0
\end{pmatrix}\lambda$} & \parbox{1.5cm}{ $\epsilon_2=\pm 1$,  $\lambda>0$}
  & \parbox{2.5cm}{(positive/negative helicity) lightlike screw} & \parbox{1.7 cm}{oriented timelike 2-plane} \\
 \\
(p8) & {\footnotesize$\begin{pmatrix}
0 & 0  & -1 & 0 & 0  \\
0 & 0 & -1 & 0 & 0 \\
 -1 & 1 & 0 & 0 & 0 \\
 0 & 0 & 0 & 0 & 0 \\
 0 & 0 & 0 & 0 & 0
\end{pmatrix}$} & \parbox{1.5cm}{[none]}
  & \parbox{1.5cm}{ [none]} & \parbox{1.7 cm}{oriented lightlike 2-plane and f.d.\ lightlike vector on it} \\
 \\
 \hline
\end{tabular} }
\label{table4c}
\end{table}

\clearpage

%\bibliography{../../bibliografie/simultaneity,../../bibliografie/libri,../../bibliografie/miei,../../bibliografie/mieiPreprints,../../bibliografie/mieiProceedings}
%\bibliographystyle{cmp}

\end{document}